\documentclass[english, letterpaper, 11pt]{article}

\usepackage{setspace}
\usepackage[T1]{fontenc}
\usepackage[utf8]{inputenc}
\usepackage{babel}
\usepackage{blindtext}
\usepackage{cite}
\usepackage[margin=1in]{geometry}
\usepackage[margin=1in]{caption}
\usepackage[pdftex]{graphicx}
\graphicspath{{./fig/}}
\usepackage[cmex10]{amsmath}
\usepackage{amssymb}
\usepackage{array}
\usepackage[tight,footnotesize]{subfigure}
\usepackage{stfloats}
\usepackage{algorithm,algorithmic}
\usepackage{enumerate}
\usepackage{amsthm}
\usepackage{color}
\usepackage{mathrsfs}
\usepackage{multicol}
\usepackage{mathtools}
\usepackage{tikz}
\usepackage{pgfplots}
\usepackage{umoline}
\usepackage{arydshln}
\usepackage{hyperref}
\usepackage{authblk}

\providecommand{\keywords}[1]{\bigskip\textbf{\textit{Index terms---}} #1}

\newcommand{\eqdef}{=\vcentcolon}
\newcommand{\defeq}{\vcentcolon=}
\newcommand{\norm}[1]{\left\|{#1}\right\|}

\newcommand{\bbR}{\mathbb{R}}
\newcommand{\bbC}{\mathbb{C}}
\newcommand{\bbN}{\mathbb{N}}
\newcommand{\bbZ}{\mathbb{Z}}
\newcommand{\bbP}{\mathbb{P}}

\newcommand{\calX}{\mathcal{X}}
\newcommand{\calY}{\mathcal{Y}}

\newcommand{\calM}{\mathcal{M}}

\newcommand{\calU}{\mathcal{U}}
\newcommand{\calV}{\mathcal{V}}

\newcommand{\calA}{\mathcal{A}}
\newcommand{\calB}{\mathcal{B}}
\newcommand{\real}{\operatorname{Re}}
\newcommand{\imag}{\operatorname{Im}}
\newcommand{\rmd}{\mathrm{d}}
\newcommand{\ind}[1]{\mathbf{1}\left(#1\right)}

\newtheorem{theorem}{Theorem}[section]
\newtheorem{lemma}[theorem]{Lemma}
\newtheorem{corollary}[theorem]{Corollary}

\newtheorem{definition}[theorem]{Definition}

\begin{document}
\title{Identifiability and Stability in Blind Deconvolution \\under Minimal Assumptions\thanks{This work was supported in part by the National Science Foundation (NSF) under Grants CCF 10-18789 and IIS 14-47879.}}
\author[1]{Yanjun Li}
\author[2]{Kiryung Lee}
\author[1]{Yoram Bresler}
\affil[1]{Coordinated Science Laboratory and Department of Electrical and Computer Engineering}
\affil[2]{Coordinated Science Laboratory and Department of Statistics}
\affil[1,2]{University of Illinois, Urbana-Champaign}
\date{}
\maketitle

\doublespacing

\abstract
Blind deconvolution (BD) arises in many applications. Without assumptions on the signal and the filter, BD does not admit a unique solution. In practice, subspace or sparsity assumptions have shown the ability to reduce the search space and yield the unique solution. However, existing theoretical analysis on uniqueness in BD is rather limited. In an earlier paper, we provided the first algebraic sample complexities for BD that hold for almost all bases or frames. We showed that for BD of a pair of vectors in $\bbC^n$, with subspace constraints of dimensions $m_1$ and $m_2$, respectively, a sample complexity of $n\geq m_1m_2$ is sufficient. This result is suboptimal, since the number of degrees of freedom is merely $m_1+m_2-1$. We provided analogus results, with similar suboptimality, for BD with sparsity or mixed subspace and sparsity constraints. In this paper, taking advantage of the recent progress on the information-theoretic limits of unique low-rank matrix recovery, we finally bridge this gap, and derive an optimal sample complexity result for BD with generic bases or frames. We show that for BD of an arbitrary pair (resp. all pairs) of vectors in $\bbC^n$, with sparsity constraints of sparsity levels $s_1$ and $s_2$, a sample complexity of $n > s_1+s_2$ (resp. $n > 2(s_1+s_2)$) is sufficient. We also present analogous results for BD with subspace constraints or mixed constraints, with the subspace dimension replacing the sparsity level. Last but not least, in all the above scenarios, if the bases or frames follow a probabilistic distribution specified in the paper, the recovery is not only unique, but also stable against small perturbations in the measurements, under the same sample complexities.


\section{Introduction}
Blind deconvolution (BD) is the bilinear inverse problem of recovering the signal and the filter simultaneously given the their convolutioin or circular convolution. It arises in many applications, including blind image deblurring \cite{Kundur1996}, blind channel equalization \cite{LitwinL.R.1999}, speech dereverberation \cite{Naylor2010}, and seismic data analysis \cite{Yilmaz2001}. Without further assumptions, BD is an ill-posed problem, and does not yield a unique solution. In this paper, we focus on subspace or sparsity assumptions on the signal and the filter. These priors, which render BD better-posed by reducing the search space, were shown to be effective constraints or regularizers in various applications\cite{Chan1998,Herrity2008b,Asif2009,Krishnan2011,Repetti2015,Ahmed2014}. However, despite the success in practice, the theoretical results on uniqueness in BD with a subspace or sparsity constraint are limited.

Recently, the ``lifting'' scheme -- recasting bilinear or quadratic inverse problems, such as blind deconvolution and phase retrieval, as rank-$1$ matrix recovery from linear measurements -- has attracted considerable attention \cite{Ahmed2014,Candes2013a}. Choudhary and Mitra \cite{Choudhary2013a} showed that identifiability in BD (or in any bilinear inverse problem) hinges on the set of rank-$2$ matrices in a certain nullspace. In particular, they showed a negative result that the solution to blind deconvolution with a canonical sparsity prior, that is, sparsity over the natural basis, is \emph{not} identifiable \cite{Choudhary2014a}. However, the authors did not analyze the identifiability of signals that are sparse over other dictionaries. Eldar et al. \cite{Eldar2012} derived tight sufficient conditions for low-rank matrix recovery. However, the authors did not exploit any sparsity priors, and the results do not apply to structured measurements that arise in blind deconvolution. 

Using the lifting framework, Ahmed et al. \cite{Ahmed2014}, Ling and Strohmer \cite{Ling2015}, and Lee et al. \cite{Lee2015a,Lee2015} proposed algorithms to solve BD with with subspace constraints, mixed constraints, and sparsity constraints, respectively. Chi \cite{Chi2015} solved BD with mixed constraints, where the sparse spikes do not necessarily lie on a grid. They all showed successful recovery using convex programming or alternating minimization, which implies identifiability and stability. These results are constructive, being demonstrated by establishing performance guarantees of algorithms. However, the guarantees are only shown to hold with high probability. The probability of failure is \emph{nonzero}, and decays in a power-law form as the size of the problem increases.

In earlier papers \cite{Li2015b,Li2015c}, we addressed the identifiability up to scaling in single channel blind deconvolution under subspace or sparsity constraints. We presented the first algebraic sample complexities for BD with fully deterministic signal models. In particular, we showed that for BD of a pair of vectors in $\bbC^n$, with generic subspace constraints of dimensions $m_1$ and $m_2$, the bilinear mapping is injective if $n\geq m_1m_2$. However, the number of degrees of freedom in the unknown pair of vectors is only $m_1+m_2-1$, hence the above result is suboptimal. Similarly, the sample complexities for BD with sparsity or with mixed constraints are $n\geq 2s_1s_2$ and $n\geq 2s_1m_2$, respectively, where $s_1$ and $s_2$ denote the sparsity levels of the signal and the filter. Here the cost for the unknown support is an extra factor of 2. These results suffer from the same suboptimality as the results for the subspace constraints, in comparison to the number of degrees of freedom of the continuous-valued unknowns.

In this paper, we finally bridge this gap. We show nearly optimal sufficient conditions for identifiability and stability in blind deconvolution that match the number of degrees of freedom in the unknowns. Results are given for the cases of subspace constraints, sparsity constraints, or mixed constraints, and for complex or real signal and filter.
The results of this paper provide the first \emph{tight} sample complexity bounds, without constants or log factors, for unique and stable recovery in blind deconvolution. Such tight bounds were not achieved (either for unique or for stable recovery) in any of the previous works \cite{Ahmed2014,Ling2015,Lee2015a,Chi2015}.

The tight sample complexities in the identifiability results apply to almost all\footnote{Results of similar nature, in that they apply to ``almost all'' objects of interest, have been derived for FIR multichannel deconvolution \cite{Harikumar1998} and for low-rank matrix recovery \cite{Eldar2012}.} bases or frames. Given a sufficient number of measurements, the conditions for unique recovery are violated only on a set of Lebesgue measure zero.  In this sense, these results are deterministic, requiring no probabilistic assumptions. As an immediate corollary though, if the bases or frames are drawn from any probability distribution that is absolutely continuous with respect to the Lebesgue measure (e.g., the entries are jointly Gaussian with a non-singular covariance, or i.i.d. following a uniform distribution, etc.), then the results in this paper hold: they imply that the signal and the filter are identifiable with probability $1$, which is better than being identifiable with high probability as in previous works \cite{Ahmed2014,Ling2015,Lee2015a,Chi2015}.

The unique recovery results are complemented by matching stability results. If the bases or frames follow a distribution specified later in this paper, then under the same sample complexities as in the identifiability results, the recovery is stable with high probability against small perturbations in the measurements. In this paper, the probability of failure decays in an exponential form  as the size of the problem increases, faster than the power-law decay in previous works \cite{Ahmed2014,Ling2015,Lee2015a,Chi2015}.

One of the main technical tools for the derivation of our results are results on information-theoretic limits of low-rank matrix recovery. Inspired by the brilliant work of Riegler et al. \cite{Riegler2015} on such limits for the real matrix case, we extend the results to complex matrix recovery problems. The contribution of our extension is two-fold: (i) we provide a simpler proof that gets rid of some unnecessary technicalities; and (ii) we prove a new concentration of measure inequality that enables the extension to the complex case. These results may be of independent interest.

Although all the main results of this paper are stated and proved for 1D circular convolution, they translate to 2D or higher-dimensional circular convolutions, by replacing the 1D discrete Fourier transform (DFT) with 2D or higher-dimensional DFT's.

The rest of the paper is organized as follows. In Section \ref{sec:probstat}, we formally state the blind deconvolution problem and its connection with matrix recovery. In Section \ref{sec:main}, we state our main results for the identifiability and stability in blind deconvolution of complex signals and of real signals. In Section \ref{sec:mr}, we extend the result for real matrix recovery \cite{Riegler2015} to complex matrix recovery. We prove the main results in Section \ref{sec:proof}, and conclude the paper in Section \ref{sec:conclusions}.

\section{Problem Statement}\label{sec:probstat}
\subsection{Notations}
We use lower-case letters $x$, $y$, $z$ to denote vectors, and upper-case letters $D$ and $E$ to denote matrices. We use $F$ to denote the normalized (unitary) discrete Fourier transform (DFT) matrix. Unless otherwise stated, all vectors are column vectors. The dimensions of all vectors and matrices are made clear in the context. We use superscript letters to denote subvectors or submatrices. For example, the scalar $x^{(j)}$ represents the $j$th entry of $x$. The vector $D^{(j,:)}$ represents the $j$th row of the matrix $D$. The colon notation is borrowed from MATLAB. The transpose and conjugate transpose to a matrix $A$ are denoted by $A^T$ and $A^*$, respectively. The inner product of two matrices $A$ and $M$ are denoted by $\left<A,M\right>=\operatorname{trace}(A^*M)$. We use $\norm{\cdot}_0$ to denote the $\ell_0$ ``norm'', or number of nonzero entries. We use $\norm{\cdot}_2$ to denote the $\ell_2$ norm of a vector or the spectral norm of a matrix, and $\norm{\cdot}_\mathrm{F}$ to denote  the Frobenious norm of a matrix. We use $\odot$ to denote entrywise product. Circular convolution is denoted by $\circledast$. 

We say a subset $\Omega_\calM$ of a linear vector space is a cone, if for every $M\in\Omega_\calM$ and every $\sigma>0$, the scaled vector $\sigma M\in \Omega_\calM$. The real and imaginary parts of a complex vector are denoted by $\real(x)$ and $\imag(x)$, respectively. If $\Omega_\calX$ is a subset of $\bbC^m$, then we use $\real(\Omega_\calX) = \{\real(x): x\in\Omega_\calX\}$, and $\imag(\Omega_\calX) = \{\imag(x): x\in\Omega_\calX\}$ to denote the real and imaginary parts of $\Omega_\calX$. The unit ball in $\bbR^{m}$ (with respect to the $\ell_2$ norm) centered at the origin is denoted by $\calB_{\bbR^{m}}$. Then $x+R\calB_{\bbR^{m}}$ denotes the ball in $\bbR^{m}$ of radius $R$ centered at $x$. Similarly, the unit ball in $\bbC^{m_1\times m_2}$ (with respect to the Frobenius norm) centered at the origin is denoted by $\calB_{\bbC^{m_1\times m_2}}$. Then $M+R\calB_{\bbC^{m_1\times m_2}}$ denotes the ball in $\bbC^{m_1\times m_2}$ of radius $R$ centered at $M$. We use $V_{\bbC^{m}}(R)=\int_{R\calB_{\bbC^{m}}} ~\rmd x$ to denote the volume of a ball of radius $R$ in $\bbC^{m}$. Here, the multiple integral of a real-valued function $f(x)$ over $\Omega_\calX\subset\bbC^{m}$ is defined as the multiple integral of $f(y^{(1:m)}+\sqrt{-1}y^{(m+1:2m)})$ over $\{y\in\bbR^{2m}: y^{(1:m)}+\sqrt{-1}y^{(m+1:2m)}\in\Omega_\calX\}$.

We say a property holds for almost all vectors/matrices (generic vectors/matrices) if the property holds for all vectors/matrices but a set of Lebesgue measure zero.

\subsection{Blind Deconvolution}\label{sec:bdstat}
In this paper, we study the blind deconvolution (BD) problem with the circular convolution model. It is the joint recovery of two vectors $u_0\in\bbC^n$ and $v_0\in\bbC^n$, namely the signal and the filter,\footnote{Due to symmetry, the name ``signal'' and ``filter'' can be used interchangeably.} given their circular convolution $z=u_0\circledast v_0$, subject to subspace or sparsity constraints. The constraint sets $\Omega_\calU$ and $\Omega_\calV$ are subsets of $\bbC^n$. With these definitions, the BD problem is written as follows:
\begin{align*}
\text{Find}~~&(u,v),\\
\text{s.t.}~~&u\circledast v = z,\\
& u \in \Omega_\calU,~ v \in \Omega_\calV.
\end{align*}
We consider the following scenarios for the constraints:
\begin{enumerate}
	\item \emph{(Subspace Constraints)} The signal $u$ and the filter $v$ reside in lower-dimensional subspaces spanned by the columns of $D\in \bbC^{n\times m_1}$ and $E\in \bbC^{n\times m_2}$, respectively, with $m_1,m_2<n$. The matrices $D$ and $E$ have full column ranks. The signal $u=Dx$ for some $x\in\bbC^{m_1}$. The filter $v=Ey$ for some $y\in\bbC^{m_2}$. 
	\item \emph{(Sparsity Constraints)} The signal $u$ and the filter $v$ are sparse over given dictionaries formed by the columns of $D\in \bbC^{n\times m_1}$ and $E\in \bbC^{n\times m_2}$, with sparsity level $s_1$ and $s_2$, respectively. Here $m_1$ and $m_2$ don't have to be smaller than $n$. The matrices $D$ and $E$ are bases or frames that satisfy the spark condition \cite{Donoho2003}: the spark, namely the smallest number of columns that are linearly dependent, of $D$ (resp. $E$) is greater than $2s_1$ (resp. $2s_2$). The signal $u=Dx$ for some $x\in\bbC^{m_1}$ with $\|x\|_0\leq s_1$. The filter $v=Ey$ for some $y\in\bbC^{m_2}$ with $\|y\|_0\leq s_2$.
	\item \emph{(Mixed Constraints)} The signal $u$ is sparse over a given dictionary $D\in \bbC^{n\times m_1}$, and the filter $v$ resides in a lower-dimensional subspace spanned by the columns of $E\in \bbC^{n\times m_2}$, with $m_2<n$. The matrix $D$ satisfies the spark condition, and $E$ has full column rank. The signal $u=Dx$ for some $x\in\bbC^{m_1}$ with $\|x\|_0\leq s_1$. The filter $v=Ey$ for some $y\in\bbC^{m_2}$.\footnote{We can also consider the scenario where $u$ resides in a subspace spanned by the columns of $D$, and $v$ is sparse over $E$. By symmetry, the analysis will be almost identical, and thus omitted.}
\end{enumerate}

In all three scenarios, the vectors $x$, $y$, and $z$ reside in Euclidean spaces $\bbC^{m_1}$, $\bbC^{m_2}$ and $\bbC^{n}$. Given the measurement $z=(Dx_0)\circledast(Ey_0)$, the blind deconvolution problem can be rewritten in the following form:
\begin{align*}
\text{(BD)}\qquad\text{Find}~~&(x,y),\\
\text{s.t.}~~&(Dx)\circledast (Ey) = z,\\
& x \in \Omega_\calX,~ y \in \Omega_\calY.
\end{align*}
If $D$ and $E$ satisfy the full column rank condition or the spark condition, then 
the uniqueness of $(u,v)$ is equivalent to the uniqueness of $(x,y)$. Indeed, the full rank or spark conditions are satisfied for almost all $D$ and $E$. Therefore, the results about the recovery of $(x,y)$ in BD with generic bases or frames imply the corresponding results for $(u,v)$.
For simplicity, we will discuss problem (BD) from now on. The constraint sets $\Omega_\calX$ and $\Omega_\calY$ depend on the constraints on the signal and the filter. For subspace constraints, $\Omega_\calX$ and $\Omega_\calY$ are $\bbC^{m_1}$ and $\bbC^{m_2}$, respectively. For sparsity constraints, $\Omega_\calX$ and $\Omega_\calY$ are $\{x\in \bbC^{m_1}: \|x\|_0\leq s_1\}$ and $\{y\in \bbC^{m_2}: \|y\|_0\leq s_2\}$, respectively.

\subsection{Identifiability up to Scaling}
An important question concerning the blind deconvolution problem is to determine when it admits a unique solution. The BD problem suffers from scaling ambiguity.
For any nonzero scalar $\sigma \in \bbC$ such that $\sigma x_0 \in \Omega_\calX$ and $\frac{1}{\sigma} y_0 \in \Omega_\calY$, $(D(\sigma x_0))\circledast (E(\frac{1}{\sigma} y_0))=(Dx_0)\circledast (Ey_0)=z$. Therefore, BD does not yield a unique solution if $\Omega_\calX,\Omega_\calY$ contain such scaled versions of $x_0,y_0$ (which is the case for the subspace or sparsity constraint sets in the previous section).
Any valid definition of unique recovery in BD must address this issue. Our approach is as follows. If every solution $(x,y)$ is a scaled version of $(x_0,y_0)$, then we say that $(x_0,y_0)$ can be uniquely identified up to scaling.\footnote{Unconstrained BD suffers also from shift ambiguity. If the signal and the filter are circularly shifted by $\ell$ and $-\ell$, respectively, their circular convolution remains the same. However, the BD problem with generic basis or frames does not suffer from shift ambiguity. If the signal and the filter are shifted, then they no longer reside in the same generic subspaces, or are no longer sparse with respect to the same generic dictionaries, as before.}We also consider the case when this property is satisfied by all pairs $(x_0,y_0)$ of interest. Thus we define identifiability as follows.
\begin{definition}\label{def:identifiability}
~
\begin{enumerate}
	\item \textbf{Weak identifiability:} We say that the pair $(x_0,y_0)\in \Omega_\calX\times \Omega_\calY$, in which $x_0\neq 0$ and $y_0\neq 0$, is identifiable up to scaling, if every solution $(x,y)\in \Omega_\calX\times \Omega_\calY$ satisfies $x=\sigma x_0$ and $y=\frac{1}{\sigma} y_0$ for some nonzero $\sigma$.
	\item \textbf{Strong identifiability:} We say that the set $\Omega_\calX\times \Omega_\calY$ is identifiable up to scaling, if every pair $(x_0,y_0)\in \Omega_\calX\times \Omega_\calY$ that satisfies $x_0\neq 0$ and $y_0\neq 0$ is identifiable up to scaling.
\end{enumerate}
\end{definition}

For blind deconvolution, there exists a linear operator $\mathcal{G}_{DE}:\bbC^{m_1\times m_2}\rightarrow \bbC^n$ such that
\begin{equation}
\mathcal{G}_{DE}(xy^T)=(Dx)\circledast(Ey). \label{eq:G_DE}
\end{equation}
Given the measurement $z=\mathcal{G}_{DE}(x_0y_0^T)=(Dx_0)\circledast(Ey_0)$, one can recast the BD problem as the recovery of the rank-$1$ matrix $M_0=x_0y_0^T \in \Omega_\calM=\{xy^T:x\in \Omega_\calX,y\in \Omega_\calY\}$. Using this so-called ``lifting''\cite{Ahmed2014} procedure, the lifted BD problem has the following form:
\begin{align*}
\text{(Lifted BD)}\qquad\text{Find}~~& M,\\
\text{s.t.}~~&\mathcal{G}_{DE}(M) = z,\\
& M\in \Omega_\calM.
\end{align*}
The uniqueness of $M_0$ is equivalent to the identifiability of $(x_0,y_0)$ up to scaling. In (Lifted BD), \emph{weak identifiability} means the recovery of $M_0$ is unique, or $M_0$ is the only point in $\Omega_\calM$ that maps to $\mathcal{G}_{DE}(M_0)$. \emph{Strong identifiability} means the recovery of all matrices in $\Omega_\calM$ is unique, that is $\mathcal{G}_{DE}$ is injective on $\Omega_\calM$, i.e., there exists $\mathcal{G}_{DE}^{-1}: \mathcal{G}_{DE}(\Omega_\calM)\rightarrow \Omega_\calM$.

Since $\Omega_\calX$ and $\Omega_\calY$ are cones, the lifted constraint set $\Omega_\calM$ is also a cone. As shown later, for the linear operator $\mathcal{G}_{DE}$ and the cone constraint set $\Omega_\calM$, identifiability on $\Omega_\calM$ is essentially the same as identifiability on the constraint set restricted to the unit ball $\Omega_\calM\bigcap \calB_{\bbC^{m_1\times m_2}}$. From now on, we use the shorthand notation
\begin{equation}
\Omega_\calB \defeq \Omega_\calM\bigcap \calB_{\bbC^{m_1\times m_2}}.  \label{eq:shorthand}
\end{equation}
Hence $\sigma\Omega_\calB = \Omega_\calM\bigcap \sigma\calB_{\bbC^{m_1\times m_2}}$.

\subsection{Stable Recovery}
Noise is ubiquitous in real-world applications. In a noisy setting, the measurement in matrix recovery is $z = \mathcal{G}_{DE}(M_0)+e$, where $M_0=x_0y_0^T$ denotes the true rank-$1$ matrix, and $e$ denotes noise or other perturbation in the measurement. In order to estimate $M_0$ from the measurement $z$, we consider the following constrained least squares problem:
\begin{align*}
\text{(Noisy BD)}\qquad\underset{M}{\min .}~~& \norm{\mathcal{G}_{DE}(M)-z}_2,\\
\text{s.t.}~~& M\in \sigma\Omega_\calB,
\end{align*}
where $\sigma\Omega_\calB = \{xy^T:x\in\Omega_\calX,y\in\Omega_\calY,\norm{xy^T}_\mathrm{F}\leq \sigma\}$. For all practical purposes, the solution to a blind deconvolution problem is bounded. Therefore, we solve (Noisy BD) subject to the constraint set restricted to a ball, whose radius $\sigma$ is sufficiently large.

We introduce the following two notions of stability of recovery:
\begin{definition}\label{def:stability}
~
\begin{enumerate}
	\item \textbf{Single point stability:} We say that the recovery of $M_0\in \sigma\Omega_\calB$, using measurement operator $\mathcal{G}_{DE}$ and constraint set $\sigma\Omega_\calB$, is stable, if for all $M\in\sigma\Omega_\calB$ such that $\norm{\mathcal{G}_{DE}(M)-\mathcal{G}_{DE}(M_0)}_2\leq \delta$, we have $\norm{M-M_0}_2\leq \varepsilon$.
	\item \textbf{Uniform stability:} We say that the recovery on $\sigma\Omega_\calB$ is uniformly stable if for all $M_1,M_2\in\sigma\Omega_\calB$ that satisfy $\norm{\mathcal{G}_{DE}(M_1)-\mathcal{G}_{DE}(M_2)}_2\leq \delta$, we have $\norm{M_1-M_2}_2\leq \varepsilon$.
\end{enumerate}
In both definitions, $\varepsilon=\varepsilon(\delta)$ is a function of $\delta$ that vanishes as $\delta$ approaches $0$. 
\end{definition}

It is easy to see that the stability as defined above, would guarantee the accuracy of the constrained least squares estimation. Let $M_1=x_1y_1^T$ denote the solution to (Noisy BD). Suppose the perturbation $e$ is small, i.e., $\norm{e}_2\leq \frac{\delta}{2}$ for some small $\delta > 0$. Then the deviation of $\mathcal{G}_{DE}(M_1)$ from $\mathcal{G}_{DE}(M_0)$ is small, i.e.,
\[
\norm{\mathcal{G}_{DE}(M_1)-\mathcal{G}_{DE}(M_0)}_2 \leq \norm{\mathcal{G}_{DE}(M_1)-z}_2+\norm{z-\mathcal{G}_{DE}(M_0)}_2 \leq 2\norm{\mathcal{G}_{DE}(M_0)-z}_2 = 2\norm{e}_2\leq \delta.
\]
By the definition of single point stability (or uniform stability), we have $\norm{M_1-M_0}_2\leq \varepsilon(\delta)$, which is also a small quantity.

If the recovery of $M_0$ is stable, then for every $\varepsilon>0$, there exists $\delta>0$ such that for every $M\in\sigma\Omega_\calB$ that satisfies $\norm{\mathcal{G}_{DE}(M)-\mathcal{G}_{DE}(M_0)}_2\leq \delta$, we have $\norm{M-M_0}_2\leq \varepsilon$. If the recovery is uniformly stable on $\sigma\Omega_\calB$, then for every $\varepsilon>0$, there exists $\delta>0$ such that for all $M_1,M_2\in\sigma\Omega_\calB$ that satisfy $\norm{\mathcal{G}_{DE}(M_1)-\mathcal{G}_{DE}(M_2)}_2\leq \delta$, we have $\norm{M_1-M_2}_2\leq \varepsilon$. If $\mathcal{G}_{DE}$ satisfies \emph{strong identifiability}, i.e., $\mathcal{G}_{DE}$ is invertible when restricted to $\Omega_\calM$, then \emph{single point stability} at $M_0$ implies that $\mathcal{G}_{DE}^{-1}$ is continuous at $\mathcal{G}_{DE}(M_0)$. \emph{Finally uniform stability} on $\sigma\Omega_\calB$ implies that $\mathcal{G}_{DE}^{-1}$ is uniformly continuous on $\mathcal{G}_{DE}(\sigma\Omega_\calB)$.

Suppose $\Omega_\calM$ is a cone, and we need to evaluate stability on $\sigma\Omega_\calB =\Omega_\calM\bigcap\sigma\calB_{\bbC^{m_1\times m_2}}$. We can scale $M_0$ and the radius of the ball by $\frac{1}{\sigma}$ simultaneously. If for all $M\in\Omega_\calB$ such that $\norm{\mathcal{G}_{DE}(M)-\mathcal{G}_{DE}(\frac{M_0}{\sigma})}_2\leq \delta$, we have $\norm{M-\frac{M_0}{\sigma}}_2\leq \varepsilon(\delta)$, then for all $M\in\sigma\Omega_\calB$ such that $\norm{\mathcal{G}_{DE}(M)-\mathcal{G}_{DE}(M_0)}_2\leq \delta$, we have $\norm{M-M_0}_2\leq \sigma\varepsilon(\frac{\delta}{\sigma})$. Therefore, we only need to consider the stability of recovery on the constraint set restricted to the unit ball, $\Omega_\calB$.

In the next section, we present the main results on the identifiability and stability in blind deconvolution, i.e., the optimal sample complexities that guarantee unique and stable recovery in (Lifted BD) and (Noisy BD), respectively.

\section{Main Results}\label{sec:main}

\subsection{Identifiability Results}\label{sec:identifiability}
Subspace membership and sparsity have been used as priors in blind deconvolution for a long time. Previous works either use these priors without theoretical justification \cite{Chan1998,Herrity2008b,Asif2009,Krishnan2011,Repetti2015}, or impose probabilistic models and show successful recovery with high probability \cite{Ahmed2014,Ling2015,Lee2015a,Chi2015}. The sufficient conditions for the identifiability in BD in our prequel paper \cite{Li2015b} are (except for a special class of so-called sub-band structured signals or filters) suboptimal. In this section, we present sufficient conditions for identifiability in BD, as defined in Section \ref{def:identifiability}, with minimal assumptions. First, the weak identifiability results in the following theorem are sharp to within one sample. 

\begin{theorem}[Weak Identifiability]\label{thm:weak}
If $n>d$, then for almost all $D\in\bbC^{n\times m_1}$ and $E\in\bbC^{n\times m_2}$, the pair $(x_0,y_0)\in\Omega_\calX\times \Omega_\calY$ ($x_0\neq 0$, $y_0\neq 0$) is identifiable up to scaling. Here, $d$ is the sample complexity bound, which is $m_1+m_2$, $s_1+m_2$, and $s_1+s_2$ in the subspace, mixed, and sparsity constraints scenarios, respectively. 
\end{theorem}

The above sufficient condition is appealing since it approaches the information-theoretic limit of blind deconvolution. In BD with subspace, mixed, and sparsity constraints, the number of degrees of freedom in the unknowns is $m_1+m_2-1$, $s_1+m_2-1$, and $s_1+s_2-1$, respectively. Therefore, to within one sample difference, the sample complexity presented above is optimal.

This result is a sufficient condition for weak identifiability. Unlike our results on BD with generic bases or frames in \cite{Li2015b}, which guarantee the injectivity of the bilinear mapping of circular convolution, this result only guarantees the identifiability of one pair $(x_0,y_0)$ in the constraint set. A sufficient condition for strong identifiability, which applies uniformly to all pairs $(x_0,y_0)$ in the constraint set, is presented next. In comparison to the optimal result in Theorem \ref{thm:weak}, the cost for strong identifiability is a factor of $2$ in the sample complexity.

\begin{theorem}[Strong Identifiability]\label{thm:strong}
If $n>2d$, then for almost all $D\in\bbC^{n\times m_1}$ and $E\in\bbC^{n\times m_2}$, all pairs $(x_0,y_0)\in\Omega_\calX\times \Omega_\calY$ ($x_0\neq 0$, $y_0\neq 0$) are identifiable up to scaling. Here, $d$ is the same as in Theorem \ref{thm:weak}.
\end{theorem}

The above results hold true for almost all complex matrices $D$ and $E$. However, in many real-world applications, both the signal and the filter are real vectors. Therefore, it is worthwhile to consider the special case where $D\in\bbR^{n\times m_1}$, $E\in\bbR^{n\times m_2}$, $x\in\bbR^{m_1}$, and $y\in\bbR^{m_2}$. We show that the same sample complexities still hold in this special case.
\begin{theorem}\label{thm:real}
In the case where $D$, $E$, $x$, and $y$ are real, the sample complexities in Theorems \ref{thm:weak} and \ref{thm:strong} hold for almost all $D\in\bbR^{n\times m_1}$ and $E\in\bbR^{n\times m_2}$.
\end{theorem}

All the results in this section are proved in Section \ref{sec:proof}. They hold for almost all matrices $D$ and $E$. When the sample complexity is met, the identifiability is violated only on a set of Lebesgue measure zero in the space of matrices $D$ and $E$. Therefore, if $D$ and $E$ are drawn from a distribution that is absolutely continuous with respect to the Lebesgue measure (e.g., $D$ and $E$ are independent random matrices whose entries are i.i.d. following a Gaussian distribution), then the identifiability result holds almost surely.

\subsection{Stability Results}\label{sec:stability}
The previous section gives the sample complexities that guarantee the identifiability in BD. Next, we show that the same sample complexity can guarantee stability. Recall that $\mathcal{G}_{DE}$ and $\Omega_\calB$ are defined in \eqref{eq:G_DE} and \eqref{eq:shorthand}, respectively. Here we only consider single point stability and uniform stability on $\Omega_\calB$, which correspond to Definition \ref{def:stability} with $\sigma = 1$. As argued before, stability on $\Omega_\calB$ implies stability on an arbitrary bounded set.

\begin{theorem}\label{thm:stability}
Assume that $D\in\bbC^{n\times m_1}$ and $E\in\bbC^{n\times m_2}$ are independent random matrices, such that the random vectors $\{(FD)^{(j,:)*}\}_{j=1}^{n}$ are i.i.d. following a uniform distribution on $R\calB_{\bbC^{m_1}}$, and $\{(FE)^{(j,:)*}\}_{j=1}^{n}$ are i.i.d. following a uniform distribution on $R\calB_{\bbC^{m_2}}$. 
\begin{enumerate}
	\item If $n>d$, then, with probability at least $1-C'(\frac{\delta^2}{R^4})^{n-d}(\frac{1}{\varepsilon^2})^n$, we have single point stability on $\Omega_\calB$.
	\item If $n>2d$, then, with probability at least $1-C''(\frac{\delta^2}{R^4})^{n-2d}(\frac{1}{\varepsilon^2})^n$, we have uniform stability on $\Omega_\calB$.
\end{enumerate}
Here, $d$ is the same sample complexity bound as in Theorems \ref{thm:weak} and \ref{thm:strong}. Except for a log factor, $C'$ and $C''$ only depend on $n$, $m_1$, $m_2$, $s_1$, and $s_2$. Define $C = 648~ m_1m_2 \left(1+2\ln\frac{2\sqrt{n}R^2}{3\delta}\right)$. The explicit expressions for $d$, $C'$, and $C''$ in the scenarios of subspace, mixed, or sparsity constraints are summarized in Table \ref{tab:constants}.
\end{theorem}

\begin{table}[htbp]%
\begin{center}
\begin{tabular}{ l >{\centering\arraybackslash}m{1in} >{\centering\arraybackslash}m{1.2in} >{\centering\arraybackslash}m{1.2in} >{\centering\arraybackslash}m{0in} }
\hline
& $d$ & $C'$ & $C''$ &\\
\hline
Subspace constraints & $m_1+m_2$ & $\frac{C^n}{n^{n-d}}$ & $\frac{(4C)^n}{n^{n-2d}}$ &\\[10pt]
Mixed constraints & $s_1+m_2$ & ${m_1\choose s_1}^2\frac{C^n}{n^{n-d}}$ & ${m_1\choose s_1}^4\frac{(4C)^n}{n^{n-2d}}$ &\\[10pt]
Sparsity constraints & $s_1+s_2$ & ${m_1\choose s_1}^2{m_2\choose s_2}^2\frac{C^n}{n^{n-d}}$ & ${m_1\choose s_1}^4{m_2\choose s_2}^4\frac{(4C)^n}{n^{n-2d}}$ &\\[10pt]
\hline
\end{tabular}
\end{center}
\caption{A summary of the constants in Theorem \ref{thm:stability} .}
\label{tab:constants}
\end{table}

The stability results of Theorem \ref{thm:stability} correspond to the identifiability results for the complex case, in Theorems \ref{thm:weak} and \ref{thm:strong}. Similar stability results can be derived for the case where $D$, $E$, $x$, and $y$ are real, which correspond to the identifiability results in Theorem \ref{thm:real}. They are omitted here for brevity.

In the discussion below, we interpret the single point stability result in Theorem \ref{thm:stability}. The uniform stability result can be interpreted similarly. Here, to make sure that the probability of stable recovery $1-C'(\frac{\delta^2}{R^4})^{n-d}(\frac{1}{\varepsilon^2})^{n}$ is non-trivial, let $\varepsilon=\varepsilon(\delta) > {C'}^{\frac{1}{2n}}\left(\frac{\delta}{R^2}\right)^\alpha$, where $\alpha = 1-\frac{d}{n} \in (0,1)$, and $\varepsilon(\delta)$ vanishes as $\delta$ approaches $0$.

Reconstruction signal-to-noise ratio (RSNR) and measurement signal-to-noise ratio (MSNR) are defined by:
\[
\mathrm{RSNR} = \frac{\|M_0\|_2^2}{\|M-M_0\|_2^2},\qquad \mathrm{MSNR} = \frac{\norm{\mathcal{G}_{DE}(M_0)}_2^2}{\norm{\mathcal{G}_{DE}(M)-\mathcal{G}_{DE}(M_0)}_2^2}.
\]
Consider the case when the error bounds are tight: $\norm{M-M_0}_2 = \varepsilon$, and $\norm{\mathcal{G}_{DE}(M)-\mathcal{G}_{DE}(M_0)}_2= \delta$.
Since the matrix $M_0$ resides in the unit ball, RSNR is on the order of $\frac{1}{\varepsilon^2}$. Since $\{(FD)^{(j,:)*}\}_{j=1}^{n}$ and $\{(FE)^{(j,:)*}\}_{j=1}^{n}$ are uniformly distributed on balls of radius $R$, the norm of the measurement $\mathcal{G}_{DE}(M_0)$ is on the order of $R^2$. Hence MSNR is on the order of $\frac{R^4}{\delta^2}$. Theorem \ref{thm:stability} can then be interpreted as follows: the probability of failure (unstable reconstruction) is roughly $\mathrm{RSNR}^{n}\cdot\mathrm{MSNR}^{-(n-d)}$.

Let $\varepsilon(\delta) = {C'}^{\frac{1}{2n}}\left(\frac{\delta}{R^2}\right)^{\frac{\alpha}{2}}$, where $\alpha = 1-\frac{d}{n}$, then the probability of single point stability in Theorem \ref{thm:stability} reduces to $1-(\frac{\delta}{R^2})^{n-d}$. If $n>d$, then as $\delta$ approaches $0$, the recovery error $\varepsilon(\delta)$ vanishes, and the probability $1-(\frac{\delta}{R^2})^{n-d}$ converges to $1$. This means that if $D$ and $E$ are random with the distributions specified in Theorem \ref{thm:stability}, then the recovery of $M_0$ is unique with probability $1$, which is also a corollary of Theorem \ref{thm:weak}.

Next, we establish stability for the special case where the operator $\mathcal{G}_{DE}$ is an isometry in the mean. Given any matrix $M=xy^T$, we have
\[
\mathcal{G}_{DE}^*\mathcal{G}_{DE}(M) = n\sum_{j=1}^{n} (FD)^{(j,:)*}(FD)^{(j,:)}M(FE)^{(j,:)T}\overline{(FE)^{(j,:)}},
\]
the expectation of which is
\begin{align*}
\mathbb{E}\left[\mathcal{G}_{DE}^*\mathcal{G}_{DE}(M)\right] = & \frac{n^2}{m_1m_2} \mathbb{E}\left[\norm{(FD)^{(j,:)*}}_2^2\right] \cdot \mathbb{E}\left[\norm{(FE)^{(j,:)*}}_2^2\right] M \\
= & \frac{n^2}{m_1m_2}\cdot \frac{m_1 R^2}{m_1+2}\cdot \frac{m_2 R^2}{m_2+2} M.
\end{align*}
The first line follows from the fact that the distribution of $(FD)^{(j,:)*}$ and $(FE)^{(j,:)*}$ are independent and isotropic. The second line is due to the fact that $(FD)^{(j,:)*}$ and $(FE)^{(j,:)*}$ are uniformly distributed on $R\calB_{\bbC^{m_1}}$ and $R\calB_{\bbC^{m_2}}$, respectively. It follows that by setting $R= \left(\frac{(m_1+2)(m_2+2)}{n^2}\right)^\frac{1}{4}$, we have $\mathbb{E}\left[\mathcal{G}_{DE}^*\mathcal{G}_{DE}(M)\right] = M$. 

Next, as an example, we analyze the uniform stability of the subspace constraints scenario, with this special choice of $R$. This will provide insight into how the constants vary with $n$, $m_1$, and $m_2$. Let $\varepsilon(\delta) = 2{C''}^{\frac{1}{2n}}\left(\frac{\delta}{R^2}\right)^\beta$, where $\beta = 1-\frac{2(m_1+m_2)}{n}$. Substituting the expressions for $R$ and $C''$, and ignoring the log factor, we have $\varepsilon(\delta) = O\left((m_1m_2)^\frac{1-\beta}{2}n^\frac{\beta}{2}\delta^\beta \right)$. By Theorem \ref{thm:stability}, in the subspace constraints scenario, if $n> 2(m_1+m_2)$, i.e., $\beta\in(0,1)$, then with probability at least $1-0.25^n$, we have
\[
\norm{M_1-M_2}_2 \lesssim (m_1m_2)^\frac{1-\beta}{2}n^\frac{\beta}{2}\norm{\mathcal{G}_{DE}(M_1)-\mathcal{G}_{DE}(M_2)}_2^\beta, \qquad \forall M_1,M_2\in\Omega_\calB.
\]
Hence, $\mathcal{G}_{DE}^{-1}$ is H\"{o}lder continuous of order $\beta$ on $\mathcal{G}_{DE}(\Omega_\calB)$.

We conclude this section by emphasizing the differences between the identifiability results in Section \ref{sec:identifiability} and the stability results in Section \ref{sec:stability}: 
\begin{enumerate}
	\item The identifiability results address the identifiability on cone constraint sets, whereas the stability results address the stability on the same constraint sets restricted to a ball of an arbitrary but finite radius. From a practical point of view, because the radius can be arbitrarily large, this restriction is of no significant consequence.
	\item The identifiability results hold for generic (Lebesgue almost all) matrices $D$ and $E$. The stability results hold with high probability when $D$ and $E$ follow some specific distributions.
\end{enumerate}


\section{Identifiability in Low-Rank Matrix Recovery}\label{sec:mr}
Using the lifted formulation, blind deconvolution with subspace or sparsity constraints has been reduced to the recovery, subject to constraints, of a rank-$1$ matrix from linear measurements that have a particular structure. The identifiability question in BD is thus reduced to identifiabilty in the latter recovery problem. In this section we address the more general question of identifiability in low rank matrix recovery. Our results express the sample complexity for identifiability in terms of the Minkowski dimension of the set in which the matrix to be recovered lives. These results are applied to the BD problem in Section \ref{sec:proof} to derive the main results of this paper.  

Recently, Riegler et al. \cite{Riegler2015} derived sample complexity results for low-rank matrix recovery, and for the recovery of matrices of low description complexity, that match the number of degrees of freedom. They considered the case where the matrices are real. Define the measurement operator $\calA: \bbR^{m_1\times m_2}\rightarrow \bbR^n$ as
\[
z = \calA(M_0)=\left[\left<A_1,M_0\right>,\left<A_2,M_0\right>,\cdots,\left<A_n,M_0\right>\right]^T\in\bbR^n,
\]
where the measurement matrices $A_j\in\bbR^{m_1\times m_2}$ ($j=1,2,\cdots,n$). Denoting by $\Omega_\calM\subset\bbR^{m_1\times m_2}$ the constraint set (which is assumed to be nonempty and bounded) for the unknown matrix, the matrix recovery problem is:
\begin{align*}
\text{(MR)}\qquad\text{Find}~~& M,\\
\text{s.t.}~~&\calA(M) = z,\\
& M\in \Omega_\calM.
\end{align*}
The conditions for unique solution to the matrix recovery problem (MR) are expressed in terms of the Minkowski dimension of the constraint set $\Omega_\calM$, which is defined as follows.
\begin{definition}\label{def:minkowski}
The lower and upper Minkowski dimensions of the nonempty bounded set $\Omega_\calM\subset\bbR^{m_1\times m_2}$ are
\[
\underline{\dim}_\mathrm{B}(\Omega_\calM)\defeq \underset{\rho\rightarrow 0}{\lim\inf}\frac{\log N_{\Omega_\calM}(\rho)}{\log\frac{1}{\rho}},\qquad
\overline{\dim}_\mathrm{B}(\Omega_\calM)\defeq \underset{\rho\rightarrow 0}{\lim\sup}\frac{\log N_{\Omega_\calM}(\rho)}{\log\frac{1}{\rho}},
\]
where $N_{\Omega_\calM}(\rho)$ denotes the covering number of $\Omega_\calM$ given by
\[
N_{\Omega_\calM}(\rho) = \min\biggl\{k\in\bbN: \Omega_\calM\subset \underset{i\in\{1,2,\cdots,k\}}{\bigcup} (M_i+\rho\calB_{\bbR^{m_1\times m_2}}),~M_i\in\bbR^{m_1\times m_2}\biggr\}.
\]
If $\underline{\dim}_\mathrm{B}(\Omega_\calM) = \overline{\dim}_\mathrm{B}(\Omega_\calM)$, then it is simply denoted by $\dim_\mathrm{B}(\Omega_\calM)$.
\end{definition}

The Minkowski dimension of the constraint set $\Omega_\calM$ can be used to represent its description complexity.
Riegler et al. showed that the solution to (MR) is unique if the sample complexity is greater than the description complexity. For almost all measurement matrices $A_1,A_2,\cdots,A_n\in\bbR^{m_1\times m_2}$, the recovery of $M_0\in\Omega_\calM$ is unique if $n>\underline{\dim}_\mathrm{B}(\Omega_\calM)$ (see \cite[Theorem 1]{Riegler2015}). An even more amazing result is that the same sample complexity can be achieved by rank-$1$ measurement matrices. For almost all $a_j\in\bbR^{m_1}$ and $b_j\in\bbR^{m_2}$ ($j=1,2,\cdots,n$), the recovery of $M_0\in\Omega_\calM$ from measurements $\left<a_jb_j^T,M_0\right>=a_j^TM_0b_j$ ($j=1,2,\cdots,n$) is unique if $n>\underline{\dim}_\mathrm{B}(\Omega_\calM)$ (see \cite[Theorem 2 and Lemma 3]{Riegler2015}). 

We state and prove the extension of this result to the case where the matrices are complex. The Minkowski dimension of the constraint set of complex matrices $\Omega_\calM\subset\bbC^{m_1\times m_2}$ can be defined as in Definition \ref{def:minkowski}, with the real number field $\bbR$ replaced by the complex number field $\bbC$. As will be shown in the next section, by simply changing the number field from real to complex, the Minkowski dimension of a set doubles. Meanwhile, by taking $n$ complex-valued measurements, the number of real-valued measurements also doubles (from $n$ to $2n$). Theorem \ref{thm:extension} shows that, together with the fact that the Minkowski dimension doubles for the complex case, we need the same number of complex-valued measurements in complex matrix recovery as we need real-valued measurements in real matrix recovery.

\begin{theorem}\label{thm:extension}
Suppose the set $\Omega_\calM\subset\bbC^{m_1\times m_2}$ is non-empty and bounded. For almost all sets of vectors $a_j\in\bbC^{m_1}$ and $b_j\in\bbC^{m_2}$ ($j=1,2,\cdots,n$), there does not exist a matrix $M\in \Omega_\calM\backslash \{0\}$ such that $\left<a_jb_j^T,M\right>=a_j^*M\overline{b_j}=0$ for $j=1,2,\cdots,n$, if $2n>\underline{\dim}_\mathrm{B}(\Omega_\calM)$.
\end{theorem}

\begin{proof}
We prove Theorem \ref{thm:extension} using the following lemma.
\begin{lemma} \label{cla:probability}
Suppose the set $\Omega_\calM\subset\bbC^{m_1\times m_2}$ is non-empty and bounded. Let the vectors $\{a_j\}_{j=1}^{n}$ and $\{b_j\}_{j=1}^{n}$ be independent random vectors, where $\{a_j\}_{j=1}^{n}$ are i.i.d. following a uniform distribution on $R\calB_{\bbC^{m_1}}$, and $\{b_j\}_{j=1}^{n}$ are i.i.d. following a uniform distribution on $R\calB_{\bbC^{m_2}}$. If $2n>\underline{\dim}_\mathrm{B}(\Omega_\calM)$, then 
\[
P \defeq \bbP\left[\exists M\in \Omega_\calM\backslash \{0\},\text{ s.t. } a_j^*M\overline{b_j}=0\text{ for }j=1,2,\cdots,n\right] = 0.
\]
\end{lemma}

We use $\mathscr{N}(\Omega,\{a_j\}_{j=1}^{n},\{b_j\}_{j=1}^{n})$ to denote the event that there exists $M\in\Omega$ such that $a_j^*M\overline{b_j}=0$ for $j=1,2,\cdots,n$. Restricted to the same support $R\calB_{\bbC^{m_1}}\times R\calB_{\bbC^{m_2}}$, the Lebesgue measure is absolutely continuous with respect to the uniform distribution. (Because the uniform measure is also absolutely continuous with respect to the Lebesgue measure, the two measures are equivalent.) If the probability of the event $\mathscr{N}(\Omega_\calM\backslash \{0\},\{a_j\}_{j=1}^{n},\{b_j\}_{j=1}^{n})$ is zero, then the Lebesgue measure of the set of $\{a_j\}_{j=1}^{n}$ and $\{b_j\}_{j=1}^{n}$, over which the event happens, is zero too. It follows that, for Lebesgue almost all $a_j\in R\calB_{\bbC^{m_1}}$ and $b_j\in R\calB_{\bbC^{m_2}}$ ($j=1,2,\cdots,n$), the event $\mathscr{N}(\Omega_\calM\backslash \{0\},\{a_j\}_{j=1}^{n},\{b_j\}_{j=1}^{n})$ does not happen. This argument is true for arbitrary radius $R$. Hence if $2n>\underline{\dim}_\mathrm{B}(\Omega_\calM)$, then by Lemma \ref{cla:probability} the event $\mathscr{N}(\Omega_\calM\backslash \{0\},\{a_j\}_{j=1}^{n},\{b_j\}_{j=1}^{n})$ does not happen, and therefore this event does not happen for almost all $a_j\in\bbC^{m_1}$ and $b_j\in\bbC^{m_2}$ ($j=1,2,\cdots,n$), i.e., there does not exist a matrix $M\in \Omega_\calM\backslash \{0\}$ such that $a_j^*M\overline{b_j}=0$ for $j=1,2,\cdots,n$. Therefore, we only need to prove Lemma \ref{cla:probability}, thus completing the proof of Theorem \ref{thm:extension}.
\end{proof}

\begin{proof}[Proof of Lemma \ref{cla:probability}]
The set $\Omega_\calM\backslash \{0\}$ can be written as
\[
\Omega_\calM\backslash \{0\} = \bigcup_{L\in\bbZ^+}\Omega_{\calM,L},
\]
where $\Omega_{\calM,L}\defeq\{M\in\Omega_\calM: \frac{1}{L}\leq \norm{M}_2\leq L\}$. By a union bound, we have
\[
P \leq \sum_{L\in\bbZ^+} P_L,
\]
where
\[
P_L \defeq \bbP\left[\exists M\in \Omega_{\calM,L}, \text{ s.t. } a_j^*M\overline{b_j}=0\text{ for }j=1,2,\cdots,n\right].
\]
In order to show that $P =0$, it suffices to prove that $P_L=0$ for all $L\in\bbZ^+$.

Let $L$ be an arbitrary positive integer. We form a minimal cover of $\Omega_{\calM,L}$ with balls of radius $\rho$ centered at the points $\{M_{\rho,L,i}\}_{i=1}^{N_{\Omega_{\calM,L}}(\rho)}$. These points may or may not be in $\Omega_{\calM,L}$. However, by the minimality of the cover, the intersection of $\Omega_{\calM,L}$ with each ball is nonempty, hence there exists another set of points $\{M'_{\rho,L,i}\}_{i=1}^{N_{\Omega_{\calM,L}}(\rho)}$ such that
\[
M'_{\rho,L,i} \in \Omega_{\calM,L} \bigcap (M_{\rho,L,i}+\rho\calB_{\bbC^{m_1\times m_2}}), \quad i = 1,2,\cdots,N_{\Omega_{\calM,L}}(\rho).
\]
Now we cover $\Omega_{\calM,L}$ with balls of radius $2\rho$ centered at $\{M'_{\rho,L,i}\}_{i=1}^{N_{\Omega_{\calM,L}}(\rho)}$, which are points in $\Omega_{\calM,L}$  (a property that will be needed for inequality \eqref{eq:useconcentration} below), because
\[
(M_{\rho,L,i}+\rho\calB_{\bbC^{m_1\times m_2}}) \subset (M'_{\rho,L,i}+2\rho\calB_{\bbC^{m_1\times m_2}}), \quad i = 1,2,\cdots,N_{\Omega_{\calM,L}}(\rho),
\]
\[
\Omega_{\calM,L} \subset \bigcup_{1\leq i\leq N_{\Omega_{\calM,L}}(\rho)} (M_{\rho,L,i}+\rho\calB_{\bbC^{m_1\times m_2}}) \subset \bigcup_{1\leq i\leq N_{\Omega_{\calM,L}}(\rho)} (M'_{\rho,L,i}+2\rho\calB_{\bbC^{m_1\times m_2}}).
\]
Defining $\delta = R^2\rho$, we have
\begin{align}
P_L \leq & \sum_{i=1}^{N_{\Omega_{\calM,L}}(\rho)} \bbP\left[\exists M\in (M'_{\rho,L,i}+2\rho\calB_{\bbC^{m_1\times m_2}}), \text{ s.t. } a_j^*M\overline{b_j}=0\text{ for }j=1,2,\cdots,n\right] \label{eq:union}\\
\leq & \sum_{i=1}^{N_{\Omega_{\calM,L}}(\rho)} \bbP\left[\exists M\in (M'_{\rho,L,i}+2\rho\calB_{\bbC^{m_1\times m_2}}), \text{ s.t. } \left|a_j^*M\overline{b_j}\right|\leq \delta \text{ for }j=1,2,\cdots,n\right] \label{eq:relax}\\
\leq & \sum_{i=1}^{N_{\Omega_{\calM,L}}(\rho)} \bbP\left[\left|a_j^*M'_{\rho,L,i}\overline{b_j}\right|\leq 3\delta \text{ for }j=1,2,\cdots,n\right] \label{eq:triangle}\\
= & \sum_{i=1}^{N_{\Omega_{\calM,L}}(\rho)} \prod_{j=1}^{n}\bbP\left[\left|a_j^*M'_{\rho,L,i}\overline{b_j}\right|\leq 3\delta\right] \label{eq:independent}\\
\leq & N_{\Omega_\calM}\left(\frac{\delta}{R^2}\right) ~(3\delta)^{2n} g(3\delta,\frac{1}{L},L,R)^n. \label{eq:useconcentration}
\end{align}
Inequality \eqref{eq:union} uses a union bound. The event in \eqref{eq:union} implies the event in \eqref{eq:relax}, which then implies the event in \eqref{eq:triangle}. Inequality \eqref{eq:triangle} is due to the following chain of inequalities, of which the last is implied by $|a_j^*M\overline{b_j}|\leq \delta$:
\begin{align*}
\left|a_j^*M'_{\rho,L,i}\overline{b_j}\right| \leq & \left|a_j^*(M'_{\rho,L,i}-M)\overline{b_j}\right| + \left|a_j^*M\overline{b_j}\right| \\
\leq & \norm{a_j}_2\norm{M'_{\rho,L,i}-M}_2\norm{b_j}_2 + \left|a_j^*M\overline{b_j}\right|\\
\leq & \norm{a_j}_2\norm{M'_{\rho,L,i}-M}_\mathrm{F}\norm{b_j}_2 + \left|a_j^*M\overline{b_j}\right|\\
\leq & 2R^2\rho +\delta ~=~ 3\delta.
\end{align*}
Equation \eqref{eq:independent} is due to the independence between random vector pairs $\{a_j,b_j\}_{j=1}^{n}$.
Inequality \eqref{eq:useconcentration} uses the fact that $N_{\Omega_{\calM,L}}(\rho) \leq N_{\Omega_\calM}(\rho)=N_{\Omega_\calM}\left(\frac{\delta}{R^2}\right)$, and the concentration of measure inequality $\bbP\left[\left|a_j^*M'_{\rho,L,i}\overline{b_j}\right|\leq \delta \right]\leq \delta^2g(\delta,\frac{1}{L},L,R)$ in Lemma \ref{lem:concentration2} in Appendix \ref{app:concentration}. (By construction, $M'_{\rho, L,i}$, as points in $\Omega_{\calM,L}$, satisfy the norm bounds $\frac{1}{L}\leq \norm{M'_{\rho, L,i}}_2\leq L$.) Here $g(\delta,\frac{1}{L},L,R)$ is a function of $\delta$ defined in \eqref{eq:gexpression} in Appendix \ref{app:concentration} , which satisfies $\lim\limits_{\delta\rightarrow 0} \frac{\log g(\delta,\frac{1}{L},L,R)}{\log \frac{1}{\delta}} = 0$.

Next, we show that \eqref{eq:useconcentration} implies $P_L = 0$. Assume the contrary, i.e. $P_L > 0$. Since $P_L$ does not depend on $\delta$, we have $\underset{\delta\rightarrow 0}{\lim\inf}\frac{\log P_L}{\log\frac{1}{\delta}} = 0$. By \eqref{eq:useconcentration} and the assumed sample complexity $2n>\underline{\dim}_\mathrm{B}(\Omega_\calM)$,
\[0 = \underset{\delta\rightarrow 0}{\lim\inf}\frac{\log P_L}{\log\frac{1}{\delta}}
\leq \underset{\delta\rightarrow 0}{\lim\inf}\frac{\log N_{\Omega_\calM}\left(\frac{\delta}{R^2}\right) + 2n\log (3\delta)  + n\log g(3\delta,\frac{1}{L},L,R)}{\log\frac{1}{\delta}}
= \underline{\dim}_\mathrm{B}(\Omega_\calM) -2n <  0,
\]
which is a contradiction. Since $L$ is arbitrary, we have $P_L=0$ for all $L\in\bbZ^+$. This completes the proof of Lemma \ref{cla:probability}.
\end{proof}

Corollaries \ref{cor:weak_bounded} and \ref{cor:strong_bounded} are direct consequences of Theorem \ref{thm:extension}.

\begin{corollary}[Weak Identifiability, Bounded]\label{cor:weak_bounded} Suppose the constraint set $\Omega_\calM \subset\bbC^{m_1\times m_2}$ is nonempty and bounded. For almost all $a_j\in\bbC^{m_1}$ and $b_j\in\bbC^{m_2}$ ($j=1,2,\cdots,n$), the recovery of $M_0$ from measurements $\left<a_jb_j^T,M_0\right>$ ($j=1,2,\cdots,n$) is unique if $2n>\underline{\dim}_\mathrm{B}(\Omega_\calM)$.
\end{corollary}

\begin{proof}
Define the set $\Omega_\calM-M_0 = \{M_1-M_0 | M_1\in\Omega_\calM\}$.
Saying that the recovery of $M_0$ from $a_j^*M_0\overline{b_j}$ ($j=1,2,\cdots,n$) is unique, is equivalent to saying that there does not exist a matrix $M$ in $(\Omega_\calM - M_0)\backslash \{0\}$ such that $\left<a_jb_j^T,M\right>=0$ ($j=1,2,\cdots,n$).

Since the set $\Omega_\calM-M_0$ is the shift of the set $\Omega_\calM$ by $M_0$, we have that $\underline{\dim}_\mathrm{B}(\Omega_\calM-M_0)=\underline{\dim}_\mathrm{B}(\Omega_\calM)$. Therefore, Corollary \ref{cor:weak_bounded} follows from Theorem \ref{thm:extension}.
\end{proof}

\begin{corollary}[Strong Identifiability, Bounded]\label{cor:strong_bounded}
Suppose the constraint set $\Omega_\calM \subset\bbC^{m_1\times m_2}$ is nonempty and bounded. For almost all $a_j\in\bbC^{m_1}$ and $b_j\in\bbC^{m_2}$ ($j=1,2,\cdots,n$), the recovery of all matrices $M_0\in\Omega_\calM$ from measurements $\left<a_jb_j^T,M_0\right>$ ($j=1,2,\cdots,n$) is unique if $n>\overline{\dim}_\mathrm{B}(\Omega_\calM)$.
\end{corollary}

\begin{proof}
Define the set $\Omega_\calM-\Omega_\calM = \{M_1-M_2 | M_1,M_2\in\Omega_\calM\}$. Saying that the recovery of all matrices in $\Omega_\calM$ is unique, is equivalent to saying that there does not exist a matrix $M$ in $(\Omega_\calM - \Omega_\calM)\backslash \{0\}$ such that $\left<a_jb_j^T,M\right>=0$ ($j=1,2,\cdots,n$).

By Lemma \ref{lem:setminus} in Appendix \ref{app:lemmas},
\[
\underline{\dim}_\mathrm{B}(\Omega_\calM-\Omega_\calM)\leq \overline{\dim}_\mathrm{B}(\Omega_\calM-\Omega_\calM) \leq 2 \overline{\dim}_\mathrm{B}(\Omega_\calM).
\]
Therefore, Corollary \ref{cor:strong_bounded} follows from Theorem \ref{thm:extension}.
\end{proof}

The proof of Theorem \ref{thm:extension} is very similar to the proofs of Theorem 2 and Lemma 3 in the paper by Riegler et al. \cite{Riegler2015}. In fact, the proofs in this paper can serve as a simpler proof of the sample complexity for the real matrix recovery problem, which is $n>\underline{\dim}_\mathrm{B}(\Omega_\calM)$, by making the following modifications:
\begin{enumerate}
	\item Changing the number field from complex to real.
	\item Using a different concentration of measure inequality in \eqref{eq:useconcentration}:
	\[
	\bbP\left[\left|a_j^TM'_{\rho,L,i}b_j \right|\leq \delta \right] \leq \delta f(\delta,\frac{1}{L},L,R),
	\]
	which is formally stated and proved in Lemma \ref{lem:concentration1}, where $f(\delta,\frac{1}{L},L,R)$ is a function of $\delta$ that satisfies $\lim\limits_{\delta\rightarrow 0} \frac{\log f(\delta,\frac{1}{L},L,R)}{\log \frac{1}{\delta}} = 0$. Hence $P_L \leq N_{\Omega_\calM}\left(\frac{\delta}{R^2}\right) ~(3\delta)^n f(3\delta,\frac{1}{L},L,R)^n$. If $n>\underline{\dim}_\mathrm{B}(\Omega_\calM)$, then $P_L = 0$ for all $L\in\bbZ^+$.
\end{enumerate}

Owing to the linearity of the measurements in the matrix recovery problem, the above results can be easily extended to the case where the constraint set is a cone. To avoid verbosity, we only prove Corollary \ref{cor:weak}. Corollary \ref{cor:strong} can be proved in a similar fashion. 

\begin{corollary}[Weak Identifiability, Unbounded]\label{cor:weak}
Suppose the constraint set $\Omega_\calM \subset\bbC^{m_1\times m_2}$ is a cone. For almost all $a_j\in\bbC^{m_1}$ and $b_j\in\bbC^{m_2}$ ($j=1,2,\cdots,n$), the recovery of $M_0$ from measurements $\left<a_jb_j^T,M_0\right>=a_j^*M_0\overline{b_j}$ ($j=1,2,\cdots,n$) is unique if $2n>\underline{\dim}_\mathrm{B}(\Omega_\calB)$, where $\Omega_\calB = \Omega_\calM \bigcap \calB_{\bbC^{m_1\times m_2}}$.
\end{corollary}

\begin{corollary}[Strong Identifiability, Unbounded]\label{cor:strong}
Suppose the constraint set $\Omega_\calM \subset\bbC^{m_1\times m_2}$ is a cone. For almost all $a_j\in\bbC^{m_1}$ and $b_j\in\bbC^{m_2}$ ($j=1,2,\cdots,n$), the recovery of all matrices $M_0\in\Omega_\calM$ from measurements $\left<a_jb_j^T,M_0\right>=a_j^*M_0\overline{b_j}$ ($j=1,2,\cdots,n$) is unique if $n>\overline{\dim}_\mathrm{B}(\Omega_\calB)$, where $\Omega_\calB = \Omega_\calM \bigcap \calB_{\bbC^{m_1\times m_2}}$.
\end{corollary}

\begin{proof}[Proof of Corollary \ref{cor:weak}]
We prove uniqueness by contradiction. Suppose that the recovery of $M_0$ is not unique, i.e., there exists $M_1 \in\Omega_\calM$ such that $\left<a_jb_j^T,M_1\right>=\left<a_jb_j^T,M_0\right>$ ($j=1,2,\cdots,n$). Let $\sigma\defeq 2\max\{\norm{M_0}_\mathrm{F},\norm{M_1}_\mathrm{F}\}>0$. Since $\Omega_\calM$ is a cone, we have
\[
\frac{1}{\sigma}M_0,~\frac{1}{\sigma}M_0 \in \Omega_\calB,\]
\[
\left<a_jb_j^T,\frac{1}{\sigma}M_1\right> =\left<a_jb_j^T,\frac{1}{\sigma}M_0\right>,\quad j=1,2,\cdots,n.
\]
Therefore, when the matrix recovery problem is restricted to a nonempty bounded constraint set $\Omega_\calB$, the recovery of $\frac{1}{\sigma}M_0$ is not unique. This, however, contradicts the sample complexity $2n>\underline{\dim}_\mathrm{B}(\Omega_\calB)$ and Corollary \ref{cor:weak_bounded}.
\end{proof}

Corollaries \ref{cor:weak} and \ref{cor:strong} show that the solution to the matrix recovery problem with a cone constraint set is unique, if the solution to the corresponding problem restricted to the unit ball is unique.


\section{Proof of the Main Results}\label{sec:proof}

\subsection{Proof of Theorems \ref{thm:weak} and \ref{thm:strong}}\label{sec:proof_identifiability}
The identifiability of $(x_0,y_0)$ up to scaling in (BD) is equivalent to the uniqueness of $M_0=x_0y_0^T$ in (Lifted BD). Note that
\[
z = \mathcal{G}_{DE}(M_0) = (Dx_0)\circledast(Ey_0) = \sqrt{n}F^*[(FDx_0)\odot(FEy_0)],
\]
\[
\frac{1}{\sqrt{n}}(Fz)^{(j)}=(FD)^{(j,:)}x_0(FE)^{(j,:)}y_0 = (FD)^{(j,:)}x_0y_0^T(FE)^{(j,:)T} =  a_j^*M_0\overline{b_j},
\]
where $a_j = (FD)^{(j,:)*}$ is the conjugate transpose of the $j$th row of $FD$, and $b_j=(FE)^{(j,:)*}$ is the conjugate transpose of the $j$th row of $FE$.
Rewriting (Lifted BD) in the frequency domain: 
\begin{align*}
\text{(Lifted BD)$_f$}\qquad\text{Find}~~& M,\\
\text{s.t.}~~& a_j^*M\overline{b_j}= \frac{1}{\sqrt{n}}(Fz)^{(j)},\quad j=1,2\cdots,n\\
& M\in \Omega_\calM = \{xy^T: x\in\Omega_\calX,y\in\Omega_\calY\}.
\end{align*}

Clearly, the constraint set $\Omega_\calM$ is a cone. Since $a_j = (FD)^{(j,:)*}$ and $b_j=(FE)^{(j,:)*}$, there exists a bijection between the pair $(D,E)\in\bbC^{n\times m_1}\times \bbC^{n\times m_2}$ and the set of vector pairs $\{a_j\in\bbC^{m_1},b_j\in\bbC^{m_2}\}_{j=1}^{n}$.
By Corollary \ref{cor:weak}, the recovery of $M_0$ is unique for almost all $D\in\bbC^{n\times m_1}$ and $E\in\bbC^{n\times m_2}$ if $2n > \underline{\dim}_\mathrm{B}(\Omega_\calB)$. By Corollary \ref{cor:strong}, the recovery of all matrices in $\Omega_\calM$ is unique for almost all $D\in\bbC^{n\times m_1}$ and $E\in\bbC^{n\times m_2}$ if $n > \overline{\dim}_\mathrm{B}(\Omega_\calB)$. Hence, Theorems \ref{thm:weak} and \ref{thm:strong} follow from the upper bounds on Minkowski dimensions in Lemma \ref{lem:minkowski}.

\begin{lemma}\label{lem:minkowski}
The upper Minkowski dimensions of $\Omega_\calB=\Omega_\calM\bigcap \calB_{\bbC^{m_1\times m_2}}$ in (Lifted BD) with subspace, mixed, and sparsity constraints are bounded by $2(m_1+m_2)$, $2(s_1+m_2)$, and $2(s_1+s_2)$, respectively.
\end{lemma}

\begin{proof}[Proof of Lemma \ref{lem:minkowski}]
For simplicity, we only prove the upper bound for the mixed constraint set. The bounds for the other two scenarios can be proved in a similar fashion.
First of all, 
\begin{align*}
\Omega_\calB &= \{xy^T: x\in\Omega_\calX,y\in\Omega_\calY, \norm{xy^T}_\mathrm{F}\leq 1\}\\
&= \{xy^T: x\in\Omega_\calX,y\in\Omega_\calY, \norm{x}_2\leq 1,\norm{y}_2\leq 1\}\\
&= \{xy^T: x\in\Omega_\calX\bigcap \calB_{\bbC^{m_1}},y\in\Omega_\calY\bigcap \calB_{\bbC^{m_2}}\}
\end{align*}
By Lemmas \ref{lem:product} and \ref{lem:realimag}, we have
\begin{align}
\overline{\dim}_\mathrm{B}(\Omega_\calM\bigcap \calB_{\bbC^{m_1\times m_2}}) \leq & \overline{\dim}_\mathrm{B}(\Omega_\calX\bigcap \calB_{\bbC^{m_1}})+\overline{\dim}_\mathrm{B}(\Omega_\calY\bigcap \calB_{\bbC^{m_2}}) \nonumber\\
\leq & \overline{\dim}_\mathrm{B}\left(\real\left(\Omega_\calX\bigcap \calB_{\bbC^{m_1}}\right)\right)+\overline{\dim}_\mathrm{B}\left(\imag\left(\Omega_\calX\bigcap \calB_{\bbC^{m_1}}\right)\right) \nonumber\\
& +\overline{\dim}_\mathrm{B}\left(\real\left(\Omega_\calY\bigcap \calB_{\bbC^{m_2}}\right)\right)+\overline{\dim}_\mathrm{B}\left(\imag\left(\Omega_\calY\bigcap \calB_{\bbC^{m_2}}\right)\right). \label{eq:dimsum} 
\end{align}

Recall that, in the mixed constraints scenario, the filter satisfies a subspace constraint, and $\Omega_\calY = \bbC^{m_2}$. The restriction to the unit ball is $\Omega_\calY\bigcap \calB_{\bbC^{m_2}} =\calB_{\bbC^{m_2}}$, whose real and imaginary parts are $\calB_{\bbR^{m_2}}$.
By a standard volume argument,
\begin{equation}
N_{\calB_{\bbR^{m_2}}}(\rho) \leq \left(\frac{3}{\rho}\right)^{m_2}. \label{eq:volume}
\end{equation}
Hence
\begin{align}
\overline{\dim}_\mathrm{B}\left(\real\left(\Omega_\calY\bigcap \calB_{\bbC^{m_2}}\right)\right)=& \overline{\dim}_\mathrm{B}\left(\imag\left(\Omega_\calY\bigcap \calB_{\bbC^{m_2}}\right)\right) \nonumber\\
=& \overline{\dim}_\mathrm{B}(\calB_{\bbR^{m_2}}) \nonumber\\
=& \underset{{\rho\rightarrow 0}}{\lim\sup} \frac{\log N_{\calB_{\bbR^{m_2}}}(\rho)}{\log \frac{1}{\rho}} \nonumber\\
\leq & \underset{{\rho\rightarrow 0}}{\lim\sup}~ m_2 \frac{\log \frac{3}{\rho}}{\log \frac{1}{\rho}} \nonumber\\
=& m_2. \label{eq:dimY}
\end{align}

Meanwhile, the signal satisfies a sparsity constraint, and $\Omega_\calX =\{x\in\bbC^{m_1}:\|x\|_0\leq s_1\}$. The restriction to the unit ball is $\Omega_\calX\bigcap \calB_{\bbC^{m_1}} = \{x\in\bbC^{m_1}:\|x\|_0\leq s_1, \|x\|_2\leq 1\}$, whose real and imaginary parts are
\[
\real\left(\Omega_\calX\bigcap \calB_{\bbC^{m_1}}\right)=\imag\left(\Omega_\calX\bigcap \calB_{\bbC^{m_1}}\right) = \{x\in\bbR^{m_1}:\|x\|_0\leq s_1, \|x\|_2\leq 1\},
\]
which is the union of unit balls in $s_1$-dimensional subspaces. Denote this set by $\Gamma^{m_1}_{s_1,1}$. By a standard volume argument,
\[
N_{\Gamma^{m_1}_{s_1,1}}(\rho) \leq {m_1\choose s_1}\left(\frac{3}{\rho}\right)^{s_1}\leq \left(\frac{em_1}{s_1}\right)^{s_1}\left(\frac{3}{\rho}\right)^{s_1},
\]
where the second inequality follows from Stirling's approximation. Hence,
\begin{align}
\overline{\dim}_\mathrm{B}\left(\real\left(\Omega_\calX\bigcap \calB_{\bbC^{m_1}}\right)\right)=& \overline{\dim}_\mathrm{B}\left(\imag\left(\Omega_\calX\bigcap \calB_{\bbC^{m_1}}\right)\right) \nonumber\\
=& \overline{\dim}_\mathrm{B}(\Gamma^{m_1}_{s_1,1}) \nonumber\\
=& \underset{{\rho\rightarrow 0}}{\lim\sup} \frac{\log N_{\Gamma^{m_1}_{s_1,1}}(\rho)}{\log \frac{1}{\rho}} \nonumber\\
\leq & \underset{{\rho\rightarrow 0}}{\lim\sup}~ s_1 \frac{\log \frac{1}{\rho} + \log \frac{3e m_1}{s_1}}{\log \frac{1}{\rho}} \nonumber\\
=& s_1. \label{eq:dimX}
\end{align}
Combining \eqref{eq:dimsum}, \eqref{eq:dimY}, and \eqref{eq:dimX}, we have that the upper Minkowski dimension of the mixed constraint set is bounded by $2(s_1+m_2)$.
\end{proof}

\subsection{Proof of Theorem \ref{thm:real}}
Next, we prove Theorem \ref{thm:real}, which establishes results corresponding to those of Theorems \ref{thm:weak} and \ref{thm:strong} in the case where $D$, $E$, $x$, and $y$ are real. When $D$ are $E$ are real matrices, $a_j = (FD)^{(j,:)*}$ and $b_j=(FE)^{(j,:)*}$ are complex vectors, but they are no longer generic. Therefore, Corollaries \ref{cor:weak} and \ref{cor:strong} cannot be applied directly to this case.

\begin{proof}[Proof of Theorem \ref{thm:real}]
By \eqref{eq:dimsum} in the proof of Theorem \ref{lem:minkowski}, when $x$, $y$, and $M=xy^T$ are real, the Minkowski dimensions of the restricted constraint sets are half those in Theorem \ref{lem:minkowski}. For subspace, mixed, and sparsity constraints, the upper Minkowski dimensions of the restricted constraint sets are bounded by $m_1+m_2$, $s_1+m_2$, and $s_1+s_2$, respectively. To maintain the same sample complexities, we need to show a result analogous to Theorem \ref{thm:extension}, in which $a_j = (FD)^{(j,:)*}$ and $b_j=(FE)^{(j,:)*}$, $D$ and $E$ are real matrices, and $n>\underline{\dim}_\mathrm{B}(\Omega_\calM)$ is sufficient.
\begin{lemma}\label{lem:real}
Suppose $\Omega_\calM\subset\bbR^{m_1\times m_2}$ is a nonempty bounded set. Let $D\in\bbR^{n\times m_1}$ and $E\in\bbR^{n\times m_2}$, $a_j = (FD)^{(j,:)*}$ and $b_j=(FE)^{(j,:)*}$ ($j=1,2,\cdots,n$). For almost all $D\in\bbR^{n\times m_1}$ and $E\in\bbR^{n\times m_2}$, there does not exist a matrix $M\in \Omega_\calM\backslash \{0\}$ such that $\left<a_jb_j^T,M\right>=a_j^*M\overline{b_j}=0$ for $j=1,2,\cdots,n$, if $n>\underline{\dim}_\mathrm{B}(\Omega_\calM)$.
\end{lemma}
The proof of Lemma \ref{lem:real} is very similar to that of Theorem \ref{thm:extension}. In fact, the only difference is the following: the mapping between the real matrices $D,E$ and the complex vectors $\{a_j\}_{j=1}^{n}$, $\{b_j\}_{j=1}^{n}$ is no longer a bijection. The vectors $a_1$ and $b_1$ are real vectors. Due to the conjugate symmetry of DFT, the vectors $a_j$ and $a_{n+2-j}$ is a conjugate pairs, i.e. $a_j = \overline{a_{n+2-j}}$. The same is true for $b_j$ and $b_{n+2-j}$. Therefore, (roughly) the first half of the DFT measurements contain all the information of real-valued unknowns. There exists a bijection between $D,E$ and the vectors $\{a_j\}_{j=1}^{\lceil (n+1)/2 \rceil}$, $\{b_j\}_{j=1}^{\lceil (n+1)/2 \rceil}$.

Due to this subtlety, in the probabilistic argument (analogous to Lemma \ref{cla:probability}) we assume $\{a_j\}_{j=1}^{\lceil (n+1)/2 \rceil}$, $\{b_j\}_{j=1}^{\lceil (n+1)/2 \rceil}$ are i.i.d. random vectors drawn from the following distribution:
\begin{itemize}
	\item When $n$ is even, $a_1$, $a_{\frac{n}{2}+1}$ are i.i.d. real random vectors following a uniform distribution on $R\calB_{\bbR^{m_1}}$, and $b_1$, $b_{\frac{n}{2}+1}$ are i.i.d. real random vectors following a uniform distribution on $R\calB_{\bbR^{m_2}}$. The vectors $\{a_j\}_{j=2}^{\frac{n}{2}}$ are i.i.d. complex random vectors following a uniform distribution on $R\calB_{\bbC^{m_1}}$, and $\{b_j\}_{j=2}^{\frac{n}{2}}$ are i.i.d. complex random vectors following a uniform distribution on $R\calB_{\bbC^{m_2}}$. 
	\item When $n$ is odd, $a_1$ is a real random vectors following a uniform distribution on $R\calB_{\bbR^{m_1}}$, and $b_1$ is a real random vectors following a uniform distribution on $R\calB_{\bbR^{m_2}}$. The vectors $\{a_j\}_{j=2}^{\frac{n+1}{2}}$ are i.i.d. complex random vectors following a uniform distribution on $R\calB_{\bbC^{m_1}}$, and $\{b_j\}_{j=2}^{\frac{n+1}{2}}$ are i.i.d. complex random vectors following a uniform distribution on $R\calB_{\bbC^{m_2}}$. 
\end{itemize}

The proof of Lemma \ref{cla:probability} is changed correspondingly. (As before, we define $\delta = \rho R^2$.)  When bounding the probability $P_L$, \eqref{eq:independent} and \eqref{eq:useconcentration} now become:
\begin{itemize}
	\item When $n$ is even,
\begin{align*}
P_L \leq & \sum_{i=1}^{N_{\Omega_{\calM,L}}(\rho)} \prod_{j=1}^{\frac{n}{2}+1}\bbP\left[\left|a_j^*M'_{\rho,L,i}\overline{b_j}\right|\leq 3\delta\right] \\
\leq & N_{\Omega_\calM}\left(\frac{\delta}{R^2}\right) ~\bigl(3\delta f(3\delta,\frac{1}{L},L,R)\bigr)^2 \bigl((3\delta)^{2} g(3\delta,\frac{1}{L},L,R)\bigr)^{\frac{n}{2}-1} \\
= & N_{\Omega_\calM}\left(\frac{\delta}{R^2}\right) (3\delta)^n f(3\delta,\frac{1}{L},L,R)^2 g(3\delta,\frac{1}{L},L,R)^{\frac{n}{2}-1}.
\end{align*}
	\item When $n$ is odd,
\begin{align*}
P_L \leq & \sum_{i=1}^{N_{\Omega_{\calM,L}}(\rho)} \prod_{j=1}^{\frac{n+1}{2}}\bbP\left[\left|a_j^*M'_{\rho,L,i}\overline{b_j}\right|\leq 3\delta \right] \\
\leq & N_{\Omega_\calM}\left(\frac{\delta}{R^2}\right) ~\bigl(3\delta f(3\delta,\frac{1}{L},L,R)\bigr) \bigl((3\delta)^{2} g(3\delta,\frac{1}{L},L,R)\bigr)^{\frac{n-1}{2}} \\
= & N_{\Omega_\calM}\left(\frac{\delta}{R^2}\right) (3\delta)^n f(3\delta,\frac{1}{L},L,R) g(3\delta,\frac{1}{L},L,R)^{\frac{n-1}{2}}.
\end{align*}
\end{itemize}
Whether $n$ is even or odd, we have $P_L = O\bigl(N_{\Omega_\calM}\left(\frac{\delta}{R^2}\right) \delta^n\bigr)$. By the same argument as in the proof of Lemma \ref{cla:probability}, the sample complexity is $n>\underline{\dim}_\mathrm{B}(\Omega_\calM)$.
\end{proof}

\subsection{Proof of Theorem \ref{thm:stability}}\label{sec:proof_stability}

In this section, we establish the stability results in blind deconvolution. The measurement in (Noisy BD) can be rewritten in the frequency domain:
\[
\tilde{z}^{(j)}\defeq \frac{1}{\sqrt{n}}(Fz)^{(j)}=(FD)^{(j,:)}x_0(FE)^{(j,:)}y_0 + \frac{1}{\sqrt{n}}(Fe)^{(j)} = a_j^*M_0\overline{b_j}+\tilde{e}^{(j)},
\]
where $M_0=x_0y_0^T$, $a_j = (FD)^{(j,:)*}$, $b_j=(FE)^{(j,:)*}$, and $\tilde{e}= \frac{1}{\sqrt{n}}Fe$. Define linear operator $\calA(M)$ by $\calA(M)= \left[a_1^*M\overline{b_1},a_2^*M\overline{b_2},\cdots,a_n^*M\overline{b_n}\right]^T$.
We rewrite (Noisy BD) in the frequency domain:
\begin{align*}
\text{(Noisy BD)$_f$}\qquad\underset{M}{\min .}~~& \norm{\calA(M)-\tilde{z}}_2,\\
\text{s.t.}~~& M\in \sigma\Omega_\calB,
\end{align*}
where $\sigma\Omega_\calB = \{xy^T: x\in\Omega_\calX,y\in\Omega_\calY,\norm{xy^T}_\mathrm{F}\leq \sigma\}$.

Note that
\[
\calA(M) = \frac{1}{\sqrt{n}} F\mathcal{G}_{DE}(M), \quad \text{and}\quad \norm{\calA(M)}_2 = \frac{1}{\sqrt{n}} \norm{\mathcal{G}_{DE}(M)}_2.
\]
The single point stability result in the subspace constraints scenario in Theorem \ref{thm:stability} follows from Lemma \ref{lem:stability} , with every $\delta$ replaced by $\frac{\delta}{\sqrt{n}}$. All other cases can be proved using similar lemmas, which we omit here for brevity.

\begin{lemma}\label{lem:stability}
In (Noisy BD)$_f$ with subspace constraints, assume that the random vectors $\{a_j\}_{j=1}^{n}$ are i.i.d. following a uniform distribution on $R\calB_{\bbC^{m_1}}$, and $\{b_j\}_{j=1}^{n}$ are i.i.d. following a uniform distribution on $R\calB_{\bbC^{m_2}}$. Let the true matrix be $M_0\in\Omega_\calB = \Omega_\calM\bigcap \calB_{\bbC^{m_1\times m_2}}=\{xy^T:x\in\calB_{\bbC^{m_1}},y\in\calB_{\bbC^{m_2}}\}$. If $n>m_1+m_2$, then, with probability at least
\[
1-\left(648~ m_1m_2 \left(1+2\ln\frac{2R^2}{3\delta}\right)\right)^n\left(\frac{\delta^2}{R^4}\right)^{n-m_1-m_2}\left(\frac{1}{\varepsilon^2}\right)^n,
\]
for all $M\in\Omega_\calB$ such that $\norm{\calA(M)-\calA(M_0)}_2\leq \delta$, we have $\norm{M-M_0}_2\leq \varepsilon$.
\end{lemma}

\begin{proof}[Proof of Lemma \ref{lem:stability}]
We need to bound the following probability of stability:
\begin{align*}
P_s \defeq & \bbP\left[ \forall M\in\Omega_\calB, \text{ if }\norm{\calA(M)-\calA(M_0)}_2\leq \delta, \text{ then } \norm{M-M_0}_2\leq \varepsilon \right] \nonumber\\
= & 1 -\bbP\left[ \exists M\in\Omega_\calB, \text{ s.t. }\norm{\calA(M)-\calA(M_0)}_2\leq \delta, \text{ and } \norm{M-M_0}_2 > \varepsilon \right] \nonumber\\
= & 1 -\bbP\left[ \exists M\in\Omega_\calB - M_0, \text{ s.t. } \norm{M}_2 > \varepsilon \text{ and } \norm{\calA(M)}_2\leq \delta \right] \nonumber\\
\eqdef & 1- P_f,
\end{align*}
where the probability of failure $P_f$ satisfies:
\begin{align}
P_f = & \bbP\left[ \exists M\in\Omega_\calB - M_0, \text{ s.t. } \norm{M}_2 > \varepsilon \text{ and } \norm{\calA(M)}_2\leq \delta \right] \nonumber\\
\leq & \bbP\left[ \exists M\in\Omega_\calB - M_0, \text{ s.t. } \norm{M}_2 > \varepsilon \text{ and } |a_j^*M\overline{b_j}| \leq \delta, j=1,2,\cdots,n \right]  \nonumber\\
\leq & N_{\Omega_\calB}\left(\frac{\delta}{R^2}\right) ~(3\delta)^{2n} g(3\delta,\varepsilon,2,R)^n \label{eq:useusecon} \\
\leq & \left(\frac{6\sqrt{2} R^2}{\delta}\right)^{2m_1+2m_2} (3\delta)^{2n}  \left(\frac{\pi^2 \cdot V_{\bbC^{m_1-1}}(R)\cdot V_{\bbC^{m_2-1}}(R)}{\varepsilon^2\cdot V_{\bbC^{m_1}}(R)\cdot V_{\bbC^{m_2}}(R)} \left(1+2\ln\frac{2R^2}{3\delta}\right)\right)^n \label{eq:usebounds} \\
= & \left(\frac{6\sqrt{2} R^2}{\delta}\right)^{2m_1+2m_2} (3\delta)^{2n}  \left(\frac{m_1m_2}{\varepsilon^2 R^4} \left(1+2\ln\frac{2R^2}{3\delta}\right)\right)^n \label{eq:cancelR} \\
\leq & \left(\frac{R^2}{\delta}\right)^{2m_1+2m_2}(6\sqrt{2})^{2n} (3\delta)^{2n}  \left(\frac{m_1m_2}{\varepsilon^2R^4} \left(1+2\ln\frac{2R^2}{3\delta}\right)\right)^n \nonumber \\
= & \left(\frac{\delta^2}{R^4}\right)^{n-m_1-m_2}\left(\frac{1}{\varepsilon^2}\right)^n \left(648~ m_1m_2 \left(1+2\ln\frac{2R^2}{3\delta}\right)\right)^n. \nonumber
\end{align}
Inequality \eqref{eq:useusecon} follows from \eqref{eq:useconcentration}, with the norm bounds $\varepsilon< \norm{M}_2\leq 2$. In \eqref{eq:usebounds}, the bound on the covering number of $\Omega_\calB = \Omega_\calM\bigcap \calB_{\bbC^{m_1\times m_2}}=\{xy^T:x\in\calB_{\bbC^{m_1}},y\in\calB_{\bbC^{m_2}}\}$ follows from \eqref{eq:product}, \eqref{eq:realimag} in Appendix \ref{app:lemmas}, and from \eqref{eq:volume} using the following steps:
\begin{align*}
N_{\Omega_\calB}\left(\frac{\delta}{R^2}\right) \leq & N_{\calB_{\bbC^{m_1}}}\left(\frac{\delta}{2R^2}\right)N_{\calB_{\bbC^{m_2}}}\left(\frac{\delta}{2R^2}\right) \\
\leq & \left(N_{\calB_{\bbR^{m_1}}}\left(\frac{\delta}{2\sqrt{2}R^2}\right)\right)^2 \left(N_{\calB_{\bbR^{m_2}}}\left(\frac{\delta}{2\sqrt{2}R^2}\right)\right)^2 \\
\leq & \left(\left(\frac{6\sqrt{2} R^2}{\delta}\right)^{m_1}\right)^2\left(\left(\frac{6\sqrt{2} R^2}{\delta}\right)^{m_2}\right)^2 = \left(\frac{6\sqrt{2} R^2}{\delta}\right)^{2m_1+2m_2}.
\end{align*}
The expression for $g(3\delta,\frac{1}{L},2,R)$ is given by \eqref{eq:gexpression} in Appendix \ref{app:concentration}. Recall that $V_{\bbC^{m_1}}(R)$ denotes the volume of a ball of radius $R$ in $\bbC^{m_1}$. Equation \eqref{eq:cancelR} follows from the fact that $V_{\bbC^{m}}(R)=V_{\bbR^{2m}}(R) = \frac{\pi^mR^{2m}}{m!}$. That completes the proof.
\end{proof}

\section{Conclusions} \label{sec:conclusions}
We studied the identifiability of blind deconvolution problems with subspace or sparsity constraints. The sample complexity results in Section \ref{sec:main} are, to within only one sample, optimal. Our results are derived with generic bases or frames, which means they are invalid only on a set of Lebesgue measure zero. If we assume that the bases or frames are drawn from a distribution that is absolutely continuous with respect to the Lebesgue measure on the space of bases or frames, then the results hold almost surely. Furthermore, if the bases or frames follow a distribution specified in this paper, then under the same sample complexities, the recovery is not only unique with probability $1$, but also stable with high probability against small perturbations in the measurements. These results provide the first \emph{tight} sample complexity bounds, without constants or log factors, for unique or stable recovery in blind deconvolution. They are fundamental to the blind deconvolution problem, independent of algorithms.


\appendix
\renewcommand{\theequation}{\Alph{section}.\arabic{equation}}
\setcounter{equation}{0}
\section{Concentration of Measure}\label{app:concentration}
\begin{lemma}\label{lem:concentration1}
Suppose $a\in\bbR^{m_1}$ and $b\in\bbR^{m_2}$ are independent random vectors, following uniform distributions on $R\calB_{\bbR^{m_1}}$ and $R\calB_{\bbR^{m_2}}$, respectively. If a matrix $M\in\bbR^{m_1\times m_2}$ satisfies $\ell\leq \norm{M}_2\leq L$, then
\[
\bbP\left[\left|a^TM b\right|\leq \rho\right] \leq \rho f(\rho,\ell,L,R),
\]
where $f(\rho,\ell,L,R)$ satisfies $\lim_{\rho\rightarrow 0} \frac{\log f(\rho,\ell,L,R)}{\log \frac{1}{\rho}} = 0$.
\end{lemma}

\begin{proof}
Suppose the  singular value decomposition (SVD) of $M$ is
\[
M = U\Sigma V^T,
\]
where $U\in\bbR^{m_1\times m_1}$ and $V\in\bbR^{m_2\times m_2}$ are orthogonal matrices, and $\Sigma\in\bbR^{m_1\times m_2}$ satisfies $\ell<\Sigma^{(1,1)} = \norm{M}_2<L$. Let $\tilde{a} \defeq U^Ta$, and $\tilde{b} \defeq V^Tb$, then $\tilde{a}$ and $\tilde{b}$ are also independent random vectors, following uniform distributions on $R\calB_{\bbR^{m_1}}$ and $R\calB_{\bbR^{m_2}}$, respectively. 

Therefore,
\begin{align}
& \bbP\left[\left|a^TM\overline{b}\right|\leq \rho\right] \nonumber\\
= & \bbP\left[\left|\tilde{a}^T\Sigma\tilde{b}\right|\leq \rho\right] \nonumber\\
= & \frac{\int_{R\calB_{\bbR^{m_1}}} ~\rmd\tilde{a}\int_{R\calB_{\bbR^{m_2}}}~ \rmd\tilde{b} ~\ind{|\tilde{a}^T\Sigma\tilde{b}|\leq \rho} }{\int_{R\calB_{\bbR^{m_1}}} ~\rmd\tilde{a}\int_{R\calB_{\bbR^{m_2}}}~\rmd\tilde{b} } \nonumber\\
=& \frac{1}{V_{\bbR^{m_1}}(R)\cdot V_{\bbR^{m_2}}(R)} \int\limits_{R\calB_{\bbR^{m_1-1}}}  ~\rmd\tilde{a}^{(2:m_1)} \int\limits_{R\calB_{\bbR^{m_2-1}}}~\rmd\tilde{b}^{(2:m_2)}~ \phi(\tilde{a},\tilde{b}),  \label{eq:big}
\end{align}
where $V_{\bbR^{m_1}}(R)=\int_{R\calB_{\bbR^{m_1}}} ~\rmd\tilde{a}$ denotes the volume of a ball of radius $R$ in $\bbR^{m_1}$, and
\begin{align}
\phi(\tilde{a},\tilde{b}) =& \int_{-R}^{R} ~\rmd\tilde{a}^{(1)} \int_{-R}^{R} ~\rmd\tilde{b}^{(1)}~\ind{|\tilde{a}^T\Sigma\tilde{b}|\leq \rho} \cdot \ind{|\tilde{a}^{(1)}|^2 \leq R^2-\norm{\tilde{a}^{(2:m_1)}}_2^2} \cdot \ind{|\tilde{b}^{(1)}|^2 \leq R^2-\norm{\tilde{b}^{(2:m_2)}}_2^2} \nonumber\\
\leq & \int_{-R}^{R} ~\rmd\tilde{a}^{(1)} \int_{-R}^{R} ~\rmd\tilde{b}^{(1)} \ind{|\tilde{a}^{(1)}\tilde{b}^{(1)}+\frac{1}{\norm{M}_2}\tilde{a}^{(2:m_1)T}\Sigma^{(2:m_1,2:m_2)}\tilde{b}^{(2:m_2)}| \leq \frac{\rho}{\norm{M}_2}}             \label{eq:goodend}\\
\leq & \int_{-R}^{R} ~\rmd\tilde{a}^{(1)} \min\Biggl(\frac{2\rho}{\norm{M}_2|\tilde{a}^{(1)}|}, 2R \Biggl) \label{eq:badstart} \\
=  & \frac{4\rho}{\norm{M}_2} \left(1+\ln\frac{\norm{M}_2R^2}{\rho}\right) \nonumber \\
\leq & \frac{4\rho}{\ell}  \left(1+\ln\frac{LR^2}{\rho}\right). \label{eq:small}
\end{align}
Substituting \eqref{eq:small} into \eqref{eq:big}, we obtain
\[
\bbP\left[\left|a^TM\overline{b}\right|\leq \rho\right] \leq \frac{4\rho \cdot V_{\bbR^{m_1-1}}(R)\cdot V_{\bbR^{m_2-1}}(R)}{\ell\cdot V_{\bbR^{m_1}}(R)\cdot V_{\bbR^{m_2}}(R)} \left(1+\ln\frac{LR^2}{\rho}\right).
\]
Define
\[
f(\rho,\ell,L,R) \defeq \frac{4 \cdot V_{\bbR^{m_1-1}}(R)\cdot V_{\bbR^{m_2-1}}(R)}{\ell \cdot V_{\bbR^{m_1}}(R)\cdot V_{\bbR^{m_2}}(R)} \left(1+\ln\frac{LR^2}{\rho}\right).
\]
Clearly, $\lim_{\rho\rightarrow 0} \frac{\log f(\rho,\ell,L,R)}{\log \frac{1}{\rho}} = 0$.
\end{proof}

\begin{lemma}\label{lem:concentration2}
Suppose $a\in\bbC^{m_1}$ and $b\in\bbC^{m_2}$ are independent random vectors, following uniform distributions on $R\calB_{\bbC^{m_1}}$ and $R\calB_{\bbC^{m_2}}$, respectively. If  a matrix $M\in\bbC^{m_1\times m_2}$ satisfies $\ell\leq\norm{M}_2\leq L$, then
\[
\bbP\left[\left|a^*M\overline{b}\right|\leq \rho\right] \leq \rho^2 g(\rho,\ell,L,R),
\]
where $g(\rho,\ell,L,R)$ satisfies $\lim_{\rho\rightarrow 0} \frac{\log g(\rho,\ell,L,R)}{\log \frac{1}{\rho}} = 0$.
\end{lemma}
\begin{proof}
The proof follows steps mostly analogous to those in the proof of Lemma \ref{lem:concentration1} by replacing the real field by the complex field. Here, we define $\tilde{a} \defeq U^T\overline{a}$, and $\tilde{b} \defeq V^*\overline{b}$. It follows that \eqref{eq:big} -- \eqref{eq:goodend} apply, with the real field replaced by the complex field, and the interval of integration $[-R,R]$ replaced by the disk in the complex plane $R\calB_{\bbC^{1}}$. Then \eqref{eq:badstart} -- \eqref{eq:small} are replaced by
\begin{align*}
\phi(\tilde{a},\tilde{b}) \leq & \int_{R\calB_{\bbC^{1}}} ~\rmd\tilde{a}^{(1)} \min\Biggl(\frac{\pi\rho^2}{\norm{M}_2^2|\tilde{a}^{(1)}|^2}, \pi R^2\Biggl) \nonumber \\
=  & \frac{\pi^2\rho^2}{\norm{M}_2^2} \left(1+2\ln\frac{\norm{M}_2R^2}{\rho}\right) \nonumber \\
\leq & \frac{\pi^2\rho^2}{\ell^2} \left(1+2\ln\frac{LR^2}{\rho}\right). \nonumber
\end{align*}

In a manner analogous to the proof of Lemma \ref{lem:concentration1} , it follows that
\[
\bbP\left[\left|a^*M\overline{b}\right|\leq \rho\right] \leq \frac{\pi^2 \rho^2 \cdot V_{\bbC^{m_1-1}}(R)\cdot V_{\bbC^{m_2-1}}(R)}{\ell^2\cdot V_{\bbC^{m_1}}(R)\cdot V_{\bbC^{m_2}}(R)} \left(1+2\ln\frac{LR^2}{\rho}\right).
\]
Here we use $V_{\bbC^{m_1}}(R)=\int_{R\calB_{\bbC^{m_1}}} ~\rmd\tilde{a}$ to denote the volume of a ball of radius $R$ in $\bbC^{m_1}$. Define
\begin{equation}
g(\rho,\ell,L,R) \defeq \frac{\pi^2 \cdot V_{\bbC^{m_1-1}}(R)\cdot V_{\bbC^{m_2-1}}(R)}{\ell^2\cdot V_{\bbC^{m_1}}(R)\cdot V_{\bbC^{m_2}}(R)} \left(1+2\ln\frac{LR^2}{\rho}\right). \label{eq:gexpression}
\end{equation}
Clearly, $\lim_{\rho\rightarrow 0} \frac{\log g(\rho,\ell,L,R)}{\log \frac{1}{\rho}} = 0$.
\end{proof}

\setcounter{equation}{0}
\section{Useful Lemmas About Minkowski Dimension}\label{app:lemmas}

\begin{lemma}\label{lem:setminus}
Let $\Omega_\calX$ and $\Omega_\calY$ be nonempty bounded subsets of a normed vector space. Then $\overline{\dim}_\mathrm{B}(\Omega_\calX-\Omega_\calY)\leq  \overline{\dim}_\mathrm{B}(\Omega_\calX)+\overline{\dim}_\mathrm{B}(\Omega_\calY)$.
\end{lemma}

\begin{proof}
We cover $\Omega_\calX$ and $\Omega_\calY$ with balls of radius $\rho$ centered at $\{x_i\}_{i=1}^{N_{\Omega_\calX}(\rho)}$ and $\{y_i\}_{i=1}^{N_{\Omega_\calY}(\rho)}$, respectively. Given any point $x-y\in\Omega_\calX-\Omega_\calY$, we can find centers of the above covering, $x_{i_1}$ and $y_{i_2}$, such that
\[
\norm{x-x_{i_1}} \leq \rho,\quad
\norm{y-y_{i_2}} \leq \rho.
\]
Hence, 
\[
\norm{(x-y)-(x_{i_1}-y_{i_2})} \leq \norm{x-x_{i_1}}+\norm{y-y_{i_2}} \leq 2\rho.
\]
Therefore, the set $\Omega_\calX-\Omega_\calY$ can be covered by $N_{\Omega_\calX}(\rho)N_{\Omega_\calY}(\rho)$ balls of radius $2\rho$ centered at points (like $x_{i_1}-y_{i_2}$) generated by the centers $\{x_i\}_{i=1}^{N_{\Omega_\calX}(\rho)}$ and $\{y_i\}_{i=1}^{N_{\Omega_\calY}(\rho)}$. It follows that
\[
N_{(\Omega_\calX-\Omega_\calY)}(2\rho)\leq N_{\Omega_\calX}(\rho)N_{\Omega_\calY}(\rho).
\]
We then bound the Minkowski dimension:
\begin{align*}
\overline{\dim}_\mathrm{B}(\Omega_\calX-\Omega_\calY) = & \underset{\rho\rightarrow 0}{\lim\sup}\frac{\log N_{(\Omega_\calX-\Omega_\calY)}(2\rho)}{\log\frac{1}{2\rho}} \nonumber\\
\leq & \underset{\rho\rightarrow 0}{\lim\sup}\frac{\log N_{\Omega_\calX}(\rho)N_{\Omega_\calY}(\rho)}{\log\frac{1}{2\rho}} \nonumber\\
\leq & \underset{\rho\rightarrow 0}{\lim\sup} \frac{\log N_{\Omega_\calX}(\rho)}{\log\frac{1}{2\rho}}+\underset{\rho\rightarrow 0}{\lim\sup} \frac{\log N_{\Omega_\calY}(\rho)}{\log\frac{1}{2\rho}} \nonumber\\
=& \overline{\dim}_\mathrm{B}(\Omega_\calX)+\overline{\dim}_\mathrm{B}(\Omega_\calY).
\end{align*}
\end{proof}

\begin{lemma}\label{lem:product}
Let $\Omega_\calX$ and $\Omega_\calY$ be nonempty bounded subsets of $\bbC^{m_1}$ and $\bbC^{m_2}$, respectively. Let $\Omega_\calM = \{xy^T: x\in\Omega_\calX, y\in\Omega_\calY\}\subset \bbC^{m_1\times m_2}$. Then $\overline{\dim}_\mathrm{B}(\Omega_\calM)\leq  \overline{\dim}_\mathrm{B}(\Omega_\calX)+\overline{\dim}_\mathrm{B}(\Omega_\calY)$.
\end{lemma}

\begin{proof}
Since $\Omega_\calX$ and $\Omega_\calY$ are bounded, there exists a large enough constant $L$ such that
\[
\Omega_\calX \subset L\calB_{\bbC^{m_1}},\qquad \Omega_\calY \subset L\calB_{\bbC^{m_2}}.
\]
We cover $\Omega_\calX$ and $\Omega_\calY$ with balls of radius $\rho$ centered at the following two sets of points, respectively:
\[ \{x_i\}_{i=1}^{N_{\Omega_\calX}(\rho)} \subset L\calB_{\bbC^{m_1}},\qquad  \{y_i\}_{i=1}^{N_{\Omega_\calY}(\rho)} \subset L\calB_{\bbC^{m_2}}.
\]
Given any point $xy^T\in\Omega_\calM$, we can find centers of the above coverings, $x_{i_1}$ and $y_{i_2}$, such that
\[
\norm{x-x_{i_1}}_2 \leq \rho,\qquad
\norm{y-y_{i_2}}_2 \leq \rho.
\]
Then
\begin{align*}
\norm{xy^T-x_{i_1}y_{i_2}^T}_\mathrm{F} =& \norm{xy^T-x_{i_1}y^T+x_{i_1}y^T-x_{i_1}y_{i_2}^T}_\mathrm{F}\\
\leq & \norm{x-x_{i_1}}_2\norm{y}_2+\norm{y-y_{i_2}}_2\norm{x_{i_1}}_2\\
\leq & 2L\rho.
\end{align*}
Therefore, the set $\Omega_\calM$ can be covered by $N_{\Omega_\calX}(\rho)N_{\Omega_\calY}(\rho)$ balls in $\bbC^{m_1\times m_2}$ of radius $2L\rho$, centered at the rank-$1$ matrices (like $x_{i_1}y_{i_2}^T$) generated by the centers of the coverings of $\Omega_\calX$ and $\Omega_\calY$. It follows that
\begin{equation}
N_{\Omega_\calM}(2L\rho)\leq N_{\Omega_\calX}(\rho)N_{\Omega_\calY}(\rho).\label{eq:product}
\end{equation}
Therefore,
\begin{align*}
\overline{\dim}_\mathrm{B}(\Omega_\calM) =& \underset{\rho\rightarrow 0}{\lim\sup}\frac{\log N_{\Omega_\calM}(2L\rho)}{\log\frac{1}{2L\rho}}\\
\leq & \underset{\rho\rightarrow 0}{\lim\sup}\frac{\log N_{\Omega_\calX}(\rho)N_{\Omega_\calY}(\rho)}{\log\frac{1}{2L\rho}}\\
\leq & \underset{\rho\rightarrow 0}{\lim\sup}\frac{\log N_{\Omega_\calX}(\rho)}{\log\frac{1}{2L\rho}}+\underset{\rho\rightarrow 0}{\lim\sup}\frac{\log N_{\Omega_\calY}(\rho)}{\log\frac{1}{2L\rho}}\\
=& \overline{\dim}_\mathrm{B}(\Omega_\calX)+\overline{\dim}_\mathrm{B}(\Omega_\calY).
\end{align*}
\end{proof}

\begin{lemma}\label{lem:realimag}
Let $\Omega_\calX$ be a nonempty bounded subset of $\bbC^m$. Let $\real(\Omega_\calX) = \{\real(x): x\in\Omega_\calX\}$, and $\imag(\Omega_\calX) = \{\imag(x): x\in\Omega_\calX\}$. Then
$\overline{\dim}_\mathrm{B}(\Omega_\calX)\leq \overline{\dim}_\mathrm{B}(\real(\Omega_\calX))+\overline{\dim}_\mathrm{B}(\imag(\Omega_\calX))$.
\end{lemma}

\begin{proof}
The real and imaginary parts $\real(\Omega_\calX)$ and $\imag(\Omega_\calX)$ are bounded subsets of $\bbR^m$. There exists a large enough constant $L$ such that
\[
\real(\Omega_\calX), \imag(\Omega_\calX)\subset L\calB_{\bbR^m}.
\]
We cover $\real(\Omega_\calX)$ and $\imag(\Omega_\calX)$ with balls of radius $\rho$ centered at the following two sets of points, respectively:
\[ \bigl\{x_i^{\mathrm{Re}}\bigr\}_{i=1}^{N_{\real(\Omega_\calX)}(\rho)},~\bigl\{x_i^{\mathrm{Im}}\bigr\}_{i=1}^{N_{\imag(\Omega_\calX)}(\rho)}~\subset L\calB_{\bbR^m}.
\]
Given any point $x\in\Omega_\calX$, we can find centers of the above coverings, $x_{i_1}^{\mathrm{Re}}$ and $x_{i_2}^{\mathrm{Im}}$, such that
\[
\norm{\real(x)-x_{i_1}^{\mathrm{Re}}}_2\leq \rho,\qquad
\norm{\imag(x)-x_{i_2}^{\mathrm{Im}}}_2\leq \rho.
\]
Let $x_c = x_{i_1}^{\mathrm{Re}}+\sqrt{-1}x_{i_2}^{\mathrm{Im}}$. Then
\[
\norm{x-x_c}_2 = \sqrt{\norm{\real(x)-x_{i_1}^{\mathrm{Re}}}_2^2+\norm{\imag(x)-x_{i_2}^{\mathrm{Im}}}_2^2} \leq \sqrt{2}\rho
\]
Therefore, the set $\Omega_\calX$ can be covered by $N_{\real(\Omega_\calX)}(\rho)N_{\imag(\Omega_\calX)}(\rho)$ balls in $\bbC^{m}$ of radius $\sqrt{2}\rho$, centered at the complex vectors (like $x_c$) generated by the centers of the coverings of $\real(\Omega_\calX)$ and $\imag(\Omega_\calX)$. It follows that
\begin{equation}
N_{\Omega_\calX}(\sqrt{2}\rho)\leq N_{\real(\Omega_\calX)}(\rho)N_{\imag(\Omega_\calX)}(\rho). \label{eq:realimag}
\end{equation}
Therefore,
\begin{align*}
\overline{\dim}_\mathrm{B}(\Omega_\calX) 
=& \underset{\rho\rightarrow 0}{\lim\sup}\frac{\log N_{\Omega_\calX}(\sqrt{2}\rho)}{\log\frac{1}{\sqrt{2}\rho}}\\
\leq & \underset{\rho\rightarrow 0}{\lim\sup}\frac{\log N_{\real(\Omega_\calX)}(\rho)N_{\imag(\Omega_\calX)}(\rho)}{\log\frac{1}{\sqrt{2}\rho}}\\
\leq & \underset{\rho\rightarrow 0}{\lim\sup}\frac{\log N_{\real(\Omega_\calX)}(\rho)}{\log\frac{1}{\sqrt{2}\rho}}+\underset{\rho\rightarrow 0}{\lim\sup}\frac{\log N_{\imag(\Omega_\calX)}(\rho)}{\log\frac{1}{\sqrt{2}\rho}}\\
=& \overline{\dim}_\mathrm{B}(\real(\Omega_\calX))+\overline{\dim}_\mathrm{B}(\imag(\Omega_\calX)).
\end{align*}
\end{proof}


\end{document}